\documentclass[12pt,a4paper,oneside]{article}

\usepackage{amsmath,amssymb,mathtools,amsthm,amscd,rotating,array,stmaryrd,accents}
\usepackage[authoryear,round]{natbib}
\usepackage{enumitem,multirow}
\usepackage{psfrag,epsfig,boxedminipage}
\usepackage{pgfplots}
\usepackage{dsfont}
\usepackage{longtable}
\usepackage{booktabs}
\usepackage{appendix}
\pgfplotsset{compat=1.11}
\usepackage{float}

\usepackage[pdftex]{pict2e} 
\usepackage[11pt]{moresize} 
\usepackage{anyfontsize} 

\newcounter{AnnahmenCounterA}
\newcounter{AnnahmenCounterB}

\newcommand{\AnnA}{A}

\newcommand{\AnnANumm}{(\AnnA\arabic{*})}

\newtheorem*{theorem*}{Argmax Theorem}

\newtheorem{theorem}{Theorem}

\newtheorem{lemma}{Lemma}
\newtheorem{algo}{Algorithm}

\renewenvironment{proof}{{\bfseries Proof.}}

\theoremstyle{definition}

\newcommand{\D}{s}   

\DeclareMathOperator*{\argmax}{arg\,max}

\newcommand{\dt}{\,\mathrm{d}}

\newcommand{\vth}{\vartheta}
\newcommand{\fth}{\boldsymbol{\theta}}

\begin{document}

\title{Testing Truncation Dependence: The Gumbel-Barnett Copula}

\author{Anne-Marie Toparkus \& Rafael Wei{\ss}bach \\[2mm] \textit{\footnotesize{Chair of Statistics and Econometrics,}}   \\[-2mm]
        \textit{\footnotesize{  Faculty for Economic and Social Sciences,}} \\[-2mm]
        \textit{\footnotesize{University of Rostock}} \\
        }
\date{ }
\maketitle

\renewcommand{\baselinestretch}{1.5}\normalsize

\begin{abstract}
In studies on lifetimes, occasionally, the population contains statistical units that are born before the data collection has started. Left-truncated are units that deceased before this start. For all other units, the age at the study start often is recorded and we aim at testing whether this second measurement is independent of the genuine measure of interest, the lifetime. Our basic model of dependence is the one-parameter Gumbel-Barnett copula. For simplicity, the marginal distribution of the lifetime is assumed to be Exponential and for the age-at-study-start, namely the distribution of birth dates, we assume a Uniform. Also for simplicity, and to fit our application, we assume that units that die later than our study period, are also truncated. As a result from point process theory, we can approximate the truncated sample by a Poisson process and thereby derive its likelihood. Identification, consistency and asymptotic distribution of the maximum-likelihood estimator are derived. Testing for positive truncation dependence must include the hypothetical independence which coincides with the boundary of the copula's parameter space. By non-standard theory, the maximum likelihood estimator of the exponential and the copula parameter is distributed as a mixture of a two- and a one-dimensional normal distribution. For the proof, the third parameter, the unobservable sample size, is profiled out. 
An interesting result is, that it differs to view the data as truncated sample, or, as simple sample from the truncated population, but not by much. The application are 55 thousand double-truncated lifetimes of German businesses that closed down over the period 2014 to 2016. The likelihood has its maximum for the copula parameter at the parameter space boundary so that the $p$-value of test is $0.5$. The life expectancy does not increase relative to the year of foundation. Using a Farlie-Gumbel-Morgenstern copula, which models positive and negative dependence, finds that life expectancy of German enterprises even decreases significantly over time. A simulation under the condition of the application suggests that the tests retain the nominal level and have good power.   	
 \\[2mm]
\end{abstract}

\section{Introduction}

Truncation can be part of the sampling design, especially in event history analysis. Independent left and also double truncation (DT) already have induced statistical literature. Nonparametric contributions are    \cite{And0,kalblawl1989,efron1999,shen2010,doerr2019,franchae2019} and \cite{morei2021}. (Semi-)Parametric models are estimated with the likelihood, a conditional likelihood or a profile likelihood in \cite{kalblawl1989,moreira2010a,emura2017,doerr2019,doerre2020,weiswied2021} or \cite{weissbachm2021effect}.
Dependent single and double truncation are lately also studied \citep{chiou2019,emura2020,renxi2022}. Dependent truncation has some similarity with dependent censoring in that the identification of the dependency must be taken into account \cite[see e.g.][]{czado}. Retrospective sampling of lifetimes is an example of potentially dependent DT, depicted in Figure \ref{exa0}. We specify the population as all units of a kind with the birth event in a period of length $G$. 
\setlength{\unitlength}{1cm}
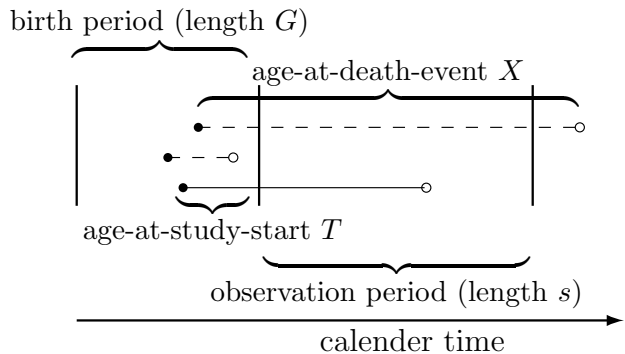
\begin{figure}[h] \centering
	\setlength{\unitlength}{0.8cm}
	\begin{picture}(10,5.5)
		\linethickness{0.3mm}
		\put(0,0.3){\vector(1,0){9.0}} \put(4.0,-0.2){calender time}
		\put(0,2.2){\line(0,1){2}}  
		\put(-1.1,4.6){$\overbrace{\vphantom{ } \hspace{2.3cm}}^{\mbox{\small birth period (length $G$)}}$}
		\put(3,2.2){\line(0,1){2}} 
		\put(7.5,2.2){\line(0,1){2}} 
		
		\linethickness{0.15mm}
		
		\put(2,3.5){\circle*{0.15}} \put(8.275,3.5){\circle{0.15}} 
		\multiput(2,3.5)(0.4,0){16}{\line(1,0){0.2}}
		
		\put(1.5,3){\circle*{0.15}} \put(2.575,3){\circle{0.15}} 
		\multiput(1.5,3)(0.4,0){3}{\line(1,0){0.2}}

		\put(1.75,2.5){\circle*{0.15}} \put(5.75,2.5){\circle{0.15}} 
		\put(1.625,2.5){ \line(1,0){3.925}}

		\put(2.0,3.8){$\overbrace{\vphantom{ } \hspace{5.0cm}}^{\mbox{\small age-at-death-event $X$}}$} 
		\put(0.1,2.4){$\underbrace{\vphantom{ } \hspace{0.96cm}}_{\mbox{\small age-at-study-start $T$}}$} 
		
		\put(2.2,1.4){$\underbrace{\vphantom{ } \hspace{3.5cm}}_{\mbox{\small observation period (length $s$)}}$} 
	\end{picture}
\caption{Three cases of the date of 1$^{st}$ event (black bullet) and date of 2$^{nd}$ event (white circle): observed (solid) and truncated (dashed) lifetimes}
\label{exa0} 
\end{figure}
We observe units affected by a death event in a period of length $s$, which for simplicity we assume to directly follow on the birth period (see Figure \ref{exa0}). Also for analytic simplicity, the lifetime $X$ is assumed to be exponentially distributed $f_E$, with parameter $\theta$.  The data must include the birthdate for each observed unit and we measure it backwards from the start of the observation period and denote this `age when the study starts' as $T$ (see Figure \ref{exa0}). For simplicity, we assume that births stem from a homogeneous Poisson process, which implies a uniform distribution $f^T$ for $T$ \cite[see][Lemma 2]{doerre2020}.
In this design, independence of the lifetime $X$ and the birthdate-equivalent $T$ contradicts the scientific consensus, at least for human mortality, of an increase in life expectancy \cite[see e.g.][]{oepvau2002}. We model the dependency with the one-dimensional parameter $\vartheta$ in the Gumbel-Barnett copula, so that, conditionally on the birthdate, the life expectancy trends upwards for the non-negative $\vartheta$, and is free of a trend for $\vartheta=0$. 

Our economic application aims at testing whether the established negative mortality trend in human demography also holds in Germany for business demography. As population we consider enterprises founded in the first quarter century after the German reunification. As data we use 55{,}279 enterprise lifetimes, double-truncated as a result of their reported closures in the period 2014 to 2016. 

In Section \ref{secpop} we define the bivariate distribution model of the population, formalize the sampling design and especially derive the selection probability and finally the likelihood. 
Section \ref{secident} studies identification as a prerequisite for the asymptotic properties and especially with the focus of the unobserved sample size as an additional parameter. Section \ref{subsecscoretest} studies the asymptotic distribution of the dependence parameter $\vartheta$ in the copula, with emphasis on the parameter space boundary. The section also includes the business demography application. Section \ref{mcs} studies the test in a Monte Carlo simulation.

\section{Population model, sampling, and likelihood} \label{secpop}

\subsection{Population and latent sample} \label{secpop2}
We consider as the population, units born within a pre-defined time window going back $G$ time units from the start of the study. The unit $i$ of the latent sample carries as a second measure to its lifetime $X_i$ ($\in \mathbb{R}_0^+$) its birthday coded as `age when the study starts' $T_i\in [0,G]$. Define $S:=\mathbb{R}_0^+ \times [0,G]$, with $0<G < \infty$, the space for one latent outcome, and let all open subsets of $S$ generate the $\sigma$-field $\mathcal{B}$.  Each unit is truncated at a different age. Let us collect notations and assumptions.
\begin{enumerate}[label=\AnnANumm]
	\item\label{A1:Compact} Parameter space: Let $\Theta:=[\varepsilon, 1 / \varepsilon] \times [0, 1- \varepsilon_{\vartheta}]$ for some `small' $\varepsilon$ and $\varepsilon_{\vartheta}>0$.
\end{enumerate}
Testing the independence hypothesis $H_0$ will coincide with the $0$ in the second dimension of $\Theta$. The statistic that indicates a deviation from $H_0$ will be the point estimate for the second parameter. Hence, for deriving its distribution under $H_0$, $\Theta$ must include the boundary $\Theta^H:=(\varepsilon, 1 / \varepsilon) \times \{0\}$. 
\begin{enumerate}[label=\AnnANumm]
	\addtocounter{enumi}{1}
	\item\label{A2:SquareInt} Marginal distributions: Let for $\theta \in [\varepsilon, 1 / \varepsilon]$, $X_i \sim Exp(\theta)$, i.e. with density $f_E(\cdot/\theta)$ and cumulative distribution function (CDF) $F_E(\cdot/\theta)$ of the exponential distribution. Let $T_i \sim Unif[0,G]$, with density $f^{T}$ and CDF $F^{T}$ of the uniform distribution.
\end{enumerate}

For two independent exponentially distributed random variables $X$ and $Y$, the bivariate survival function is $e^{-\theta_x x - \theta_y y}$, so that a simple idea of E.J. Gumbel is to model dependence by a bivariate survival function $e^{-\theta_x x - \theta_y y - \vartheta x y}$.

Our field of application are the social sciences where it is to be expected that societies in general gain progress, e.g. in their life expectation, be it of persons or of enterprises. Our main focus is a test for the hypothesis of stagnation against such progress. A suitable adaption of Gumbel's general idea, applied here to the uniform marginal distribution of $T$, is the Gumbel-Barnett copula \cite[see][Formula 2.3.5]{B3}. (We use the non-survival version of it.)  

\setcounter{AnnahmenCounterA}{5} 
\begin{enumerate}[label=\AnnANumm]
	\addtocounter{enumi}{2}
	\item\label{A3:Ind} Gumbel-Barnett copula: $(X_i,T_i)'$, $i=1,\ldots,n, n \in \mathbb{N}$, is an SRS, i.e. $\mbox{i.i.d.}$ random variables (r.v.) mapping from the probability space $(\Omega,\mathcal{A},P_{\boldsymbol{\theta}})$, with $\boldsymbol{\theta}:=(\theta,\vartheta) \in \Theta$ onto the measurable space $(S, \mathcal{B})$. $X_i$ and $T_i$ are dependent with copula
	\begin{align*}
		C^{\vth}(u,v)=u+v-1+(1-u)(1-v)e^{-\vth \log(1-u) \log(1-v)}.
	\end{align*}
\end{enumerate}
Note that $\vartheta=0$ represents independence. Scatter plots of simulated $(X_i,T_i)'$ for different $\vartheta$  in Appendix \ref{copvis} visualize the degree of negative dependence that is modelled. For large $\vth$, a small $T$, i.e. a late foundation, is associated with longer survival. The CDF of an $(X_i,T_i)'$ is 
\begin{align*}
	F_{\fth}(x,t)=\begin{cases}
		\frac{t}{G}-e^{-\theta x}+e^{-\theta x}\left(1-\frac{t}{G}\right)^{\vth \theta x+1}, & x> 0,\, 0 < t < G,\\
		1-e^{-\theta x}, & x>0, \, t\geq G,
	\end{cases}
\end{align*}
and zero elsewhere. The joint density with respect to $P_{\boldsymbol{\theta}}$
for $x >0$ and $0 < t < G$ is
\begin{equation} \label{e4}
	f_{\boldsymbol{\theta}}(x,t)=	-\frac{\theta}{G}e^{-\theta x}\left(1-\frac{t}{G}\right)^{\vth\theta x}[(\vth\theta x+1)(\vth \log(1-t/G)-1)+\vth].
\end{equation}

Kendall's tau , given as $\tau=4 \int_0^1 \int_0^1 C^{\vartheta}(u,v) d C^{\vartheta}(u,v) -1$, is a univariate measure for dependence and has no closed form but is easily seen to range from $-0.361$ and $0$. The latter is associated with $\vartheta=0$, the independence.  

To compare with, later in the data analysis, and for application in other fields, a two-sided dependence might also be of interest. \cite{gum1960} proposed a respective copula.
\begin{enumerate}
	\item[(FGM)] Farlie-Gumbel-Morgenstern copula: Let $(X_i,T_i)'$ be as in \ref{A3:Ind} and parameter space $\Theta_{FGM} := [\varepsilon, 1/\varepsilon]\times [\varepsilon_{\vartheta}-1, 1- \varepsilon_{\vartheta}]$ with copula
	\begin{align*}
		C_{FGM}^{\vth}(u,v):=u v [1+\vth^{FGM} (1-u)(1-v)].
	\end{align*}	
\end{enumerate}
Kendall's tau is given by $2 \vartheta^{FGM}/9$ and ranges between $-2/9$ and $2/9$, with independence at $\vartheta^{FGM}=0$ \cite[see][Example 5.2]{B3}. Scatter plots of simulated $(X_i,T_i)'$ for $\vth^{FGM}=\varepsilon_{\vartheta}-1 $, $\vth=0$ and $\vth^{FGM}=1- \varepsilon_{\vartheta}$ in Appendix \ref{simfgm} illustrate that, at the `extremes', dependence is weaker than for the Gumbel-Barnett copula, corresponding also to a smaller range of Kendall's tau.

\subsection{Data} \label{secpop3}

The data are a subset of the SRS in Assumption \ref{A3:Ind} of $n$ draws governed by $f_{\boldsymbol{\theta}_0}$, i.e. for the `true' parameter $\boldsymbol{\theta}_0 \in \Theta$.   
A parallelogram $D$ formalizes that a sample unit is only observed when its death falls into the observation period (of length $s$). 
\begin{enumerate}[label=\AnnANumm]
	\addtocounter{enumi}{3}
	\item\label{A4:trunc} Observation: For known constant $\D>0$, the column vector $(X_i,T_i)'$ is observed if it is in $D:=\{ (x,t)' | 0 < t \leq x \leq t + \D, t \le G \}$. 
\end{enumerate} 
Following up on \ref{A4:trunc}, we denote an {\it observation} by $(\widetilde{X}_j,\widetilde{T}_j)'$ and renumber the observed units with $j=1, \ldots, M \le n$. (Sorting the unobserved units to the end of the latent SRS is a convention already to be found in \cite{heckman1976}.) Note that $M=\sum_{i=1}^n \mathds{1}_{\{(X_i,T_i)' \in D\}}$ and is hence random. 
Now define for $\boldsymbol{\theta} \in \Theta$ and $P_{\boldsymbol{\theta}}$ from Assumption \ref{A3:Ind} the selection probability of the $i^{\mathrm{th}}$ individual
\begin{eqnarray} \label{alpha0}
	\alpha_{\boldsymbol{\theta}}  & := & P_{\boldsymbol{\theta}}\{T_i \le X_i \le T_i +s\} = \int_0^G\int_t^{t+s} f_{\fth}(x,t) \dt x \dt t = \int_D f_{\fth}(x,t)  \dt(x,t) \nonumber \\
	& = & \int_0^1 \left(c^{\vartheta}_u\{F^T[F_E^{-1}(u)]\}-c^{\vartheta}_u\{F^T[F_E^{-1}(u)]-F^T(s)\} \right)\dt u,	
\end{eqnarray}
with $c^{\vartheta}_u(v):=\partial C^{\vartheta}(u,v)/\partial u$. The selection probability is not given in closed form, but note that the numerical calculation is easy, because $D$ is bounded. Furthermore, the last expression of Equation \eqref{alpha0} is a univariate integral over a compact interval similar to \citet[][Theorem 1]{emura2020}. 
Note that, with the slight re-definition $\boldsymbol{\theta}:=(\theta, \vartheta_{FGM})'$, the selection probability based on $C_{FGM}^{\vartheta}$ is given in closed-form as:
\begin{align*}
	\alpha_{\fth}^{FGM}	&:=\frac{1}{\theta G} (1-e^{-\theta s})(1-e^{-\theta G})  \\
	& -\frac{\vth^{FGM}}{G}\left[-\frac{1}{\theta}(e^{-\theta s}-1)(e^{-\theta G}+1) 
	  +\frac{1}{2 \theta} (e^{-2\theta s}-1)(e^{-2\theta G}+1)\right] \\
	& +\frac{2 \vth^{FGM}}{G^2} \left[\frac{1}{\theta^2} (1-e^{-\theta s})(1-e^{-\theta G})-\frac{1}{4 \theta^2} (1-e^{-2\theta s})(1-e^{-2\theta G})\right]
\end{align*}
For the sake of brevity, the display of the theoretical analysis is limited here to the Gumbel-Barnett copula (i.e. Assumption (A3)), because the FGM-copula is considerably easier. However in the empirical example of Section \ref{empexa} we consider both, and compare. 

The selection probability will occur in the likelihood, so that for maximization, its first partial derivatives will be needed. The second and third partial derivatives of $\alpha_{\boldsymbol{\theta}}$ will be needed for proving the asymptotic normality and calculating the standard error. The proof of the following and explicit representations of $\alpha$'s derivatives, both needed later, are similar to those of \citet[][Corollary 1]{weiswied2021} but are omitted here.  
\begin{lemma} \label{propalpha}
	Under the Assumptions \ref{A1:Compact}-\ref{A4:trunc}  and $\boldsymbol{\theta} \in \Theta$, it is $\alpha_{\boldsymbol{\theta}} \in (0,1)$. Furthermore $\boldsymbol{\theta} \mapsto \alpha_{\boldsymbol{\theta}}$ has first, second and third partial derivatives in the directions of $\theta$ and $\vartheta$ and combinations thereof. Those derivatives are continuous in $\boldsymbol{\theta}$.
\end{lemma} 
We are now in a position to formulate the likelihood, maximize it and apply large sample theory.

\subsection{Likelihood} \label{seclike}

The likelihood springs from standard results for point processes \cite[see e.g.][Theorem 3.1.1., Section 7.1 respectively]{reiss1993,daljo}, and we maximize it later as a function of the generic $\boldsymbol{\theta}$ and $n$. (Distinguishing in notation between the true and a generic $n$ is omitted.) The idea is roughly to decompose the likelihood according to
\[
\ell^{\star}=Pr\{data\}=Pr\{(\widetilde{X}_1,\widetilde{T}_1)', \ldots, (\widetilde{X}_M,\widetilde{T}_M)'|M\}Pr\{M\}.
\] 
Note that by $Pr$, we cannot mean $P_{\boldsymbol{\theta}}$ of Assumption \ref{A3:Ind}. Detailed definitions of the measures related to the probabilities are the same as for the model with independent truncation, i.e. $\vartheta=0$,  \cite[see][]{weiswied2021}. The latter reference also proves that the $(\widetilde{X}_j,\widetilde{T}_j)'$ are stochastically independent, conditional on observation, so that $Pr\{(\widetilde{X}_1,\widetilde{T}_1)', \ldots, (\widetilde{X}_M,\widetilde{T}_M)'|M\}$ becomes a product over the conditional densities of each observation.  
With $P_{\boldsymbol{\theta}}$ from Assumption \ref{A3:Ind}, $(\widetilde{X}_j,\widetilde{T}_j)'$ has CDF
\begin{equation} \label{defxt}
	F^{\widetilde{X},\widetilde{T}}(x,t):= P_{\boldsymbol{\theta}}\left\{X_i \le x,T_i \le t|T_i \leq X_i \leq T_i + s\right\}.
\end{equation}
Leaving out the proof, an explicit relation between the distribution of a $(X_i,T_i)'$ and \eqref{defxt} is, for $(x,t)' \in D$, 
\begin{equation} \label{lemma_imeasure}
	\alpha_{\boldsymbol{\theta}} F^{\widetilde{X},\widetilde{T}}(x,t)  =\int_0^t \int_0^x f_{\fth}(y,s) dy ds - r(x,t), \quad \text{with} \quad  \frac{\partial^2}{\partial x \partial t} r(x,t) =0,
\end{equation}
under the Assumptions \ref{A1:Compact}-\ref{A4:trunc} and $\boldsymbol{\theta} \in \Theta$. The proof for the property of the remainder $r$ is similar to that in \cite{weiswied2021}. Hence the density of $(\widetilde{X}_j,\widetilde{T}_j)'$ is $f_{\fth}(x,t)/\alpha_{\boldsymbol{\theta}}$.
The  Binomial-distributed size of the observed sample, $M$, can be approximated by a Poisson-distributed $M^{\star}$, when the selection  probability $\alpha_{\boldsymbol{\theta}}$  for each of the $n$ i.i.d. Bernoulli experiments is small. This is especially the case when the width of the observation period (of length $s$) is `short', relative to the population period (of length $G$). The resulting density $Pr\{M=m^{\star}\} \approx \frac{\mu^{m^{\star}}}{m^{\star}!}e^{-\mu}$, with $\mu=n\alpha_{\boldsymbol{\theta}}$, is responsible not only for the very last (exponential) term in the following representation, but also contributes a $n^{M^{\star}}$ to the leading product.  
With $h_{\fth}(x,t):=n f_{\fth}(x,t) \mathds{1}_D(x,t)$ and using \eqref{lemma_imeasure}, the proof of \citet[][Theorem 3]{weiswied2021} extends to:
\begin{equation*}
	\ell^{\star} \approx \left(\prod_{j=1}^{M^{\star}} h_{\fth}(\widetilde{X}_j,\widetilde{T}_j)\right) e^{(G+s)G-n\alpha_{\fth}}
	=n^{M^{\star}} \left(\prod_{j=1}^{M^{\star}} f_{\fth}(\widetilde{X}_j,\widetilde{T}_j)\right) e^{(G+s)G-n\alpha_{\fth}}
\end{equation*}
Here the proximity is in the sense of a Hellinger distance. Note that $\alpha_{\boldsymbol{\theta}}$ as denominator of the density of $(\widetilde{X}_j,\widetilde{T}_j)'$ cancels the other $\alpha_{\boldsymbol{\theta}}$ in the  density of $M^{\star}$ out. Because almost surely $\widetilde{T}_j<G$ and $\ell^{\star}>0$, we have
\begin{align*}
	\log \ell^{\star} \approx \sum_{j=1}^{M^{\star}} \log n f_{\fth}(\widetilde{X}_j,\widetilde{T}_j) + (G+s)G-n\alpha_{\fth}.
\end{align*}
In this approximation $M^{\star}$ can exceed $n$, and e.g. $(\widetilde{X}_{n+1},\widetilde{T}_{n+1})'$ will not be defined. In order to guarantee that the observations fit the model (in the meaning of distributional models for point processes), we further approximate  $M^{\star} \approx M$ \cite[see again][Sect. 3]{weiswied2021}. It follows $\log \ell^{\star} \approx\log \ell$ with
\begin{align}\label{e5}
	\begin{split}
		\log \ell(\fth,n) &:=\sum_{j=1}^M \log n f_{\fth}(\widetilde{X}_j, \widetilde{T}_j) + (G+s)G-n\alpha_{\fth}\\
		&=\sum_{i=1}^{n} \mathds{1}_{\{(X_i,T_i) \in D\}}\log n f_{\fth}(X_i,T_i) + (G+s)G-n\alpha_{\fth}, 
	\end{split}
\end{align}
where again we assume strictly $\widetilde{T}_j<G$ and $T_i<G$. Profiling out $n$, we estimate the parameter $\fth$ with the first two coordinates of 
\begin{align}\label{e6}
	\argmax\limits_{\fth \in \Theta, n \in \mathbb{N}} \,\log \ell (\fth,n).
\end{align}	
Note that the maximum can be on the boundary of $\Theta$ in $\vartheta$-direction, especially in $\Theta^H$. For the score function, necessary partial derivatives with respect to $\theta$ and $\vth$ are given in closed form in Appendix \ref{score_theta_vartheta}. The latter derivatives depend on $n$.
By Appendix \ref{score_n}, $\log \ell(\fth,n)$ is maximized in $n$ by the next smallest integer to $n=M/\alpha_{\fth}$. The latter uses the fact that the logarithm can be bounded from above by a linear function and by a hyperbola from below (with proof Appendix \ref{appendixlem5} and also used in the following section).
\begin{lemma}\label{h1} For any $x>1$ holds $1-\frac{1}{x} < \log(x) < x-1$.
\end{lemma}

\section{Identification} \label{secident}

Identification is necessary to ensure the consistency of a parameter estimator. (Consistency will then be necessary for asymptotic normality.) The classic definition of identification is tailored to an SRS. An SRS is only latent in the study at hand. It will still be useful  to study SRS-identification jointly of the latent univariate model (see Assumption \ref{A2:SquareInt}) and of the dependence model (see Assumption \ref{A3:Ind}). Appendix \ref{identsrs} proves identification of $\boldsymbol{\theta}$. 

Now in the truncated sample, with an interest in inference for $\fth$, we profiled out the parameter $n$. This reduces the three-score estimating equations for \eqref{e6} (see Appendix \ref{appscore}), by solving for $n$ and inserting in the remaining two estimating equations. Instead of inserting the natural-valued solution for $n$, we use the real-valued `near-zero' $M/\alpha_{\fth}$ (see again Appendix \ref{score_n}). This `near-zero' will be covered by the applied theory from \citet[][Sect. 5.2]{vaart1998}. Specifically, insertion yields: 
\begin{subequations} \label{eq:1}
	\begin{equation}\label{e7}
		\begin{split}
			\psi_{\fth,1}(X_i,T_i)
			&:=\mathds{1}_{[T_i,T_i+s]}(X_i)\left(\frac{1}{\theta}+X_i(\vth \log(1-T_i/G)-1) \right.\\
			& \left.\quad+\frac{\vth X_i(\vth \log(1-T_i/G)-1)}{(\vth \theta X_i+1)(\vth \log(1-T_i/G)-1)+\vth}-\frac{\frac{\partial \alpha_{\fth}}{\partial \theta}}{\alpha_{\fth}}\right)
		\end{split} 
	\end{equation} \\
	\begin{equation}\label{e8}
		\begin{split}
			\psi_{\fth,2}(X_i,T_i)
			&:=\mathds{1}_{[T_i,T_i+s]}(X_i)\Bigg(\theta X_i \log(1-T_i/G) \\
			&  \quad+\frac{(2\vth \theta X_i+1)\log(1-T_i/G)-\theta X_i+1}{(\vth \theta X_i+1)(\vth \log(1-T_i/G)-1)+\vth} -\frac{\frac{\partial \alpha_{\fth}}{\partial \vartheta}}{\alpha_{\fth}}\Bigg)
		\end{split} 
	\end{equation}
\end{subequations}
With the definition $\psi_{\fth}:=(\psi_{\fth,1},\psi_{\fth,2})'$, the near-zero estimator $\hat{\fth}_n$ for the true parameter $\fth_0$ (see Figure \ref{exa}) is the zero of $\Psi_{n}(\fth):=\frac{1}{n}\sum_{i=1}^n 	\psi_{\fth}(X_i,T_i)$, if in $\Theta$, and the nearest boundary value else. Note that $\Psi_{n}(\fth)$, is observable after multiplication by $n$ and has the same zero \cite[see][Sect. 2.2]{weiswied2021}. A boundary value on $\Theta^H$ is likely under $H_0$ and will shown to have a probability of 50\%. A parameter-independent  bound for $\psi_{\boldsymbol{\theta}}$  will be needed for proving consistency and asymptotic normality.
\begin{lemma}\label{e18}
	Under Assumptions \ref{A1:Compact}-\ref{A4:trunc}, for finite constant $K_{\varepsilon}>0$, depending on $\varepsilon$ and $\varepsilon_{\vth}$, a parameter-independent bound for the norm of $\psi_{\fth}$ is $g(x,t):=\mathds{1}_{[t,t+s]}(x) \{K_{\varepsilon}+K_{\varepsilon}[1-\log(1-t/G)]\}$, namely 	
	$||\psi_{\fth}(x,t)||^2 \leq g^2(x,t)$ for all $\fth \in \Theta$.
\end{lemma}
\begin{proof}
	Essentially, bounding $\sup_{\boldsymbol{\theta} \in \Theta} \psi_{\boldsymbol{\theta},j}(X_i,T_i)$ ($j=1,2$) is enabled by using the fact that a continuous function on a compact set attains its maximum for $\alpha_{\boldsymbol{\theta}}$ and its derivative. For $\psi_{\boldsymbol{\theta},j}(X_i,T_i)$ themselves, the numerator is bounded, and the remaining $\log$-term is not bounded but shown to be integrable. Specifically, $1/\theta$ can be bounded by $1/\varepsilon$.  \qed
\end{proof}

Of course, concavity of the function to which $\Psi_{n}$ is the gradient, at $\fth_0$, will be important. Due to analytic intractability we note it as an assumption. Throughout, a dot on top of a function $\mathbb{R}^2 \to \mathbb{R}$ will signal a gradient. On top of a gradient $\mathbb{R}^2 \to \mathbb{R}^2$, it signals its Jacobi matrix, i.e. the Hessian matrix of the function.  
\begin{enumerate}[resume*]
	\item \label{A5negdef}
	Let for $\boldsymbol{\theta}_0 \in \Theta$
	\begin{align*}
		\mathbb{E}_{\fth_0}\big[\dot{\psi}_{\fth_0}(X_1,T_1)\big] =
		\left(\begin{array}{cc}
			\mathbb{E}_{\fth_0}\big[\frac{\partial}{\partial \theta} \psi_{\fth_0,1}(X_1,T_1)\big] 
			& \mathbb{E}_{\fth_0}\big[\frac{\partial}{\partial \theta} \psi_{\fth_0,2}(X_1,T_1)\big] \\
			\mathbb{E}_{\fth_0}\big[\frac{\partial}{\partial \theta} \psi_{\fth_0,2}(X_1,T_1)\big]  
			& \mathbb{E}_{\fth_0}\big[\frac{\partial}{\partial \vth} \psi_{\fth_0,2}(X_1,T_1)\big] \\
		\end{array}\right)
	\end{align*} 
	be negative definite, where here and throughout $\partial \psi_{\fth_0}/\partial \fth$ stands for $\partial \psi_{\fth}/\partial \fth |_{\fth=\fth_0}$.
\end{enumerate}
Instead of a proof, Appendix \ref{xy} plots the surface of $\mathbb{E}_{\fth_0}[\dot{\psi}_{\fth_0}(X_1,T_1)]$'s determinant, i.e. only of the even principal minor, by $\boldsymbol{\theta}_0$ on a large subset of the parameter space $\Theta$. Figure \ref{detA5graph} (left) shows that the determinant is clearly positive for a large part of the parameter space, but also reveals that in the area $\boldsymbol{\theta}_0 \in [0.01,0.02] \times [0.1,1]$ the determinant is near to zero. Figure \ref{detA5graph} (right) explores the area and shows that for $\vartheta_0>0.14$ the determinant increases again and that the minimum is attained in the range of $\theta_0 \in [0.012,0.014]$. Those latter values are inconceivable for our example of business demography, as then the life expectancy was $\approx 0.01^{-1}=100$ years. Still, even when estimation in that area will be more difficult numerically and standard errors will be larger, Assumption \ref{A5negdef} obviously holds. Demonstrating the negativity of $\mathbb{E}_{\fth_0}[\frac{\partial}{\partial \theta_0} \psi_{\fth_0,1}(X_1,T_1)] $, i.e. of the uneven principal minor is omitted here.    

We now argue that \citet[][Theorem 5.9, Condition (ii)]{vaart1998} is the relevant analogue of identification for a truncated sample.

\subsection{The profile model} \label{profmod}

Instead of truncating the sample by $D$, one can think of the data as drawing an SRS of a correspondingly truncated population, of sample size $m$. The thus defined sub-population $\widetilde{Pop}$ is depicted in Figure \ref{exa} (bottom left box).

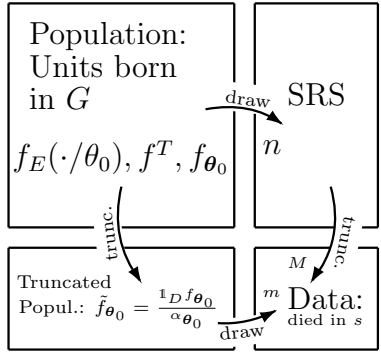
\begin{figure}[t] \centering
	\setlength{\unitlength}{0.054cm}
\begin{picture}(100,100)
	\linethickness{0.3mm}
	\put(10,35){\line(0,1){55}}
	\put(65,35){\line(0,1){55}}
	\put(10,90){\line(1,0){55}}
	\put(10,35){\line(1,0){55}}
	\put(15,80){Population:}
	\put(15,72){Units born}
	\put(15,64){in $G$}
	\put(11,50){$f_E(\cdot/\theta_0), f^T, f_{\boldsymbol{\theta}_0}$}
	
	
	\put(10,5){\line(0,1){25}}
	\put(65,5){\line(0,1){25}}
	
	\put(10,30){\line(1,0){55}}
	\put(12,20){\ssmall Truncated}	
	\put(12.5,14){\ssmall Popul.: $\tilde{f}_{\fth_0}=\frac{\mathds{1}_{D} f_{\boldsymbol{\theta}_0}}{\alpha_{\boldsymbol{\theta}_0}}$}		
	
	\put(72,53){$n$} 
	\put(70,35){\line(0,1){55}}
	\put(100,35){\line(0,1){55}}
	\put(70,35){\line(1,0){30}}
	\put(70,90){\line(1,0){30}}	
	\put(78,65){SRS}		
	
	\put(78,25){\tiny $M$}  
	\put(72,18){\tiny $m$} 
	\put(70,5){\line(0,1){25}}
	\put(100,5){\line(0,1){25}}
	\put(70,5){\line(1,0){30}}
	\put(70,30){\line(1,0){30}}
	\put(78,15){Data:}	
	\put(78,11){\tiny died in $s$}	
	
	\put(60.5,66){\rotatebox{-7}{\colorbox{white}{{{\fontsize{7}{8}\selectfont draw}}}}}
	\put(76.5,60.5){\vector(1,-0.7){0.5}}
	\qbezier(58,64)(69,66)(76.5,60.5)
	\put(30,25){\rotatebox{97}{\colorbox{white}{{{\fontsize{7}{8}\selectfont trunc.}}}}}
	\put(42,22){\vector(0.6,-1){0.5}}
	\qbezier(37.5,45)(34.5,32)(42,22) 
	\put(87,40){\rotatebox{-101}{\colorbox{white}{{{\fontsize{7}{8}\selectfont trunc.}}}}}
	\put(84,22){\vector(-0.6,-1){0.5}}
	\qbezier(88.5,45)(92,32)(84,22) 
	\put(59,7){\rotatebox{7}{\colorbox{white}{{{\fontsize{7}{8}\selectfont draw}}}}}
	\put(75,14.5){\vector(1,0.5){0.5}}
	\qbezier(61,11.5)(68,11)(75,14.5) 		
	\put(10,5){\line(1,0){55}}			
\end{picture}
	\caption{Population (with distribution of $X$, $T$ and $(X,T)'$) (top left), SRS (with sample size) (top right),  Population truncated by units died outside of the observation period (with distribution of $(\tilde X, \tilde T)'$) (bottom left), Data (as truncated SRS or SRS of truncated population) (bottom right)}
	\label{exa}
\end{figure}

This `anti-clockwise' design is not ours (defined by Assumptions \ref{A1:Compact}-\ref{A4:trunc}), but we will see that its identification enables a helpful result, also for the `clockwise' design. 
The (Lebesgue) density of the two measurements is $\tilde{f}_{\boldsymbol{\theta}}(x,t)  =  \mathds{1}_{D}(x,t) f_{\boldsymbol{\theta}}(x,t)/\alpha_{\boldsymbol{\theta}}$ and is also the density of $(\widetilde{X}_1, \widetilde{T}_1)'$.
Then, a short calculation reveals that, surprisingly, $\psi_{\boldsymbol{\theta}}$ is its score function, i.e.
$\psi_{\boldsymbol{\theta}}(x,t) = \nabla_{\boldsymbol{\theta}} \log  \tilde{f}_{\boldsymbol{\theta}}(x,t)$. (The gradient is signalled by $\nabla$, instead of a dot, when the expression is too long, as $\log  \tilde{f}$ is here.) This justifies the name profile model, and profile score for $\psi_{\boldsymbol{\theta}}$. In order to stress the implication, note that the estimators are equal for both designs. \citet[][Section 7.1, Example 7.1(a) (continued)]{daljo} reminds that this is only the case when the data can be approximated by a Poisson process, which we are allowed to, following \cite{weiswied2021}, because the selection probability $\alpha_{\boldsymbol{\theta}_0}$ will be small enough in our application.   

For a random draw $(\tilde{X}_1,\tilde{T}_1)'$ from $\widetilde{Pop}$, by Jensen's inequality, a result for the Kullback-Leibler (KL) divergence $\boldsymbol{\theta}, \, \boldsymbol{\theta}_0 \in \Theta$ is
\begin{eqnarray*}
	KL(\tilde{f}_{\boldsymbol{\theta}},\tilde{f}_{\boldsymbol{\theta}_0}) & := & \tilde{\mathbb{E}}_{\boldsymbol{\theta}_0}\left(\log \frac{\tilde{f}_{\boldsymbol{\theta}_0}(\tilde{X}_1, \tilde{T}_1)}{\tilde{f}_{\boldsymbol{\theta}}(\tilde{X}_1, \tilde{T}_1)}\right)=\int_{D} \log \frac{\tilde{f}_{\boldsymbol{\theta}_0}(x,t)}{\tilde{f}_{\boldsymbol{\theta}}(x,t)} \tilde{f}_{\boldsymbol{\theta}_0}(x,t) d(x,t) \ge  0 \\
	\text{and} & = & 0 \Leftrightarrow \boldsymbol{\theta}=\boldsymbol{\theta}_0,
\end{eqnarray*} 
if the model given by $\tilde{f}_{\boldsymbol{\theta}}$ is identified, which is shown in Appendix \ref{idcond}, and needs the mentioned latent identification (in Appendix \ref{identsrs}). 
Hence, $KL(\tilde{f}_{\boldsymbol{\theta}},\tilde{f}_{\boldsymbol{\theta}_0})$ is, as a function of $\boldsymbol{\theta}$, uniquely minimized at $\boldsymbol{\theta}=\boldsymbol{\theta}_0$. 

This is not enough for proving consistency, but enables to state now some minor consequences to be used in Section \ref{sec3_2}. The minimizing argument (and uniqueness) are unchanged when subtracting the constant $\tilde{\mathbb{E}}_{\boldsymbol{\theta}_0}[\log \tilde{f}_{\boldsymbol{\theta}_0}(\tilde{X}_1, \tilde{T}_1)]$, so that $ \tilde{\mathbb{E}}_{\boldsymbol{\theta}_0}[\log \tilde{f}_{\boldsymbol{\theta}}(\tilde{X}_1, \tilde{T}_1)]$ is uniquely maximized at $\boldsymbol{\theta}=\boldsymbol{\theta}_0$. 
As consequence of Assumption \ref{A5negdef} it is easy to verify that the (Fisher information) matrix in the profile model $\tilde{\mathbb{E}}_{\fth_0}[\dot{\psi}_{\fth_0}(\tilde{X}_1,\tilde{T}_1)]$ is negative definite. Hence, due to its given smoothness and as consequence of Assumption \ref{A5negdef}, there is a unique solution (zero) of
\begin{equation} \label{laber}
	\tilde{\mathbb{E}}_{\boldsymbol{\theta}_0}\left( \psi_{\boldsymbol{\theta}}(\tilde{X}_1,\tilde{T}_1)\right)=\nabla_{\boldsymbol{\theta}} \left[ \tilde{\mathbb{E}}_{\boldsymbol{\theta}_0}\left(\log \tilde{f}_{\boldsymbol{\theta}}(\tilde{X}_1,\tilde{T}_1)\right)\right]=0.
\end{equation}
(For that, interchange integration and differentiation similar to the proof of Lemma \ref{propalpha} \citep[and][Theorem 5.7, Chapter IV, § 5]{B8}.) 

\subsection{M-identification} \label{sec3_2}

The study of the profile model, as part of the `anti-clockwise' design, was helpful as we can now follow-up on the unique solution for \eqref{laber}. It is easy to see that, with $\mathbb{E}_{\boldsymbol{\theta}_0}$ relating to $P_{\boldsymbol{\theta}_0}$ from Assumption \ref{A3:Ind}, it is for $\boldsymbol{\theta},\, \boldsymbol{\theta}_0 \in \Theta$
\begin{equation} \label{Eumrechnung}
	\tilde{\mathbb{E}}_{\boldsymbol{\theta}_0}\left(\psi_{\boldsymbol{\theta}}(\tilde{X}_1, \tilde{T}_1)\right)= \frac{1}{\alpha_{\boldsymbol{\theta}_0}}\mathbb{E}_{\boldsymbol{\theta}_0}\left[ \psi_{\boldsymbol{\theta}}(X_1, T_1)\right].
\end{equation}
Because $\alpha_{\boldsymbol{\theta}_0}>0$, by Lemma \ref{propalpha},  
$\Psi(\boldsymbol{\theta}):=\mathbb{E}_{\boldsymbol{\theta}_0}[ \psi_{\boldsymbol{\theta}}(X_1, T_1)]=0$
also has a unique solution. 
Additionally, Appendix \ref{prooflem1} proves a result needed now (and again when we prove asymptotic normality).
\begin{lemma}\label{lem1} Under Assumptions \ref{A1:Compact}-\ref{A4:trunc}  and $\boldsymbol{\theta}_0 \in \Theta$, it is
	$\Psi(\boldsymbol{\theta}_0)=0$.
\end{lemma} 
This ends the proof of \citet[][Theorem 5.9, Condition (ii)]{vaart1998} that for any $\varepsilon > 0$ and $\boldsymbol{\theta}_0 \in \Theta$
\begin{equation*}
	\inf_{\boldsymbol{\theta}\in \Theta: d(\boldsymbol{\theta},\boldsymbol{\theta}_0) \ge \varepsilon} \Vert \Psi(\boldsymbol{\theta}) \Vert > 0= \Vert \Psi(\boldsymbol{\theta}_0) \Vert,
\end{equation*}
according to \citet[][Problem 5.27]{vaart1998}. We interpret this as an M-identification condition. The remaining Condition (i) for consistency is  convergence of $\Psi_n$ to $\Psi$ uniformly in $\boldsymbol{\theta} \in \Theta$, and will be proven in Section \ref{seccons}.

\section{Wald-type test for independence} \label{subsecscoretest}

In order to test the hypothesis of independent truncation, i.e. $H_0: \vartheta_0=0$,
for a parameter at the edge of the parameter space, a score test would be typical \cite[see e.g.][]{vossweis2014}. As an advantage, calculating the two-dimensional unrestricted estimate $\hat{\boldsymbol{\theta}}$ would not be necessary. Only the restricted one-dimensional estimator, i.e. for $\vartheta=0$, is necessary and reduces the numerical effort. And this has already been derived in \citet{weiswied2021}. However, the score asymptotically only depends on $\hat{\vartheta}$ so that we simply use the Wald-type idea to reject for a $\hat{\vartheta}$ being too large.

An important  further element will be the Fisher information, and we will need that for $\fth \in \Theta$: 
 \begin{equation} \label{deffi}
 \mathcal{I}(\boldsymbol{\theta}):=\mathbb{E}_{\fth}[\psi_{\fth}(X_i,T_i)\psi_{\fth}(X_i,T_i)'] 
\end{equation}

\subsection{Theory} \label{seccons}

 We especially need to approximate the distribution of the point estimator, the vector of zeros of \eqref{eq:1}, by a Gaussian distribution. Our data in Section \ref{empexa} will be sufficiently large to do so. We verify the classic conditions for asymptotic normality of M-estimation  \cite[see][Theorem 5.41]{vaart1998} for $\Psi_n(\fth)$, given shortly after  \eqref{eq:1}. One first condition is weak consistency. The method of proof in \citet[][Theorem 5.9]{vaart1998} even allows us to make a statement on $\Theta$, including the boundary in $\vartheta$-direction. 

\begin{theorem} \label{cons} Under assumptions \ref{A1:Compact}-\ref{A5negdef}, $\fth_0 \in \Theta$ and $\hat{\fth}_n$ defined after \eqref{eq:1}, holds $\hat{\boldsymbol{\theta}}_n  \stackrel{p}{\to} \boldsymbol{\theta}_0$ as $n \to \infty$.
\end{theorem}

\begin{proof}
	As stated above, the second condition is the content of Section \ref{secident}, when \citet[][Prob. 5.27]{vaart1998} is taken into account due to $\Theta$ being compact. In order to show the first one, we use the uniform law of large numbers \cite[see e.g.][p. 2129]{B7}. Its smoothness requirements are all fulfilled by noting that involved functions (including $\alpha_{\boldsymbol{\theta}}$ by Lemma \ref{propalpha}) are smooth, and compositions do not result in discontinuities due to division by zero. For example, for the involved term $1/\theta$, originating from the density of the exponential distribution, $\theta \ge \varepsilon>0$ by Assumption \ref{A1:Compact} avoids poles. Calculations for the denominator not to be zero are not presented here, for the sake of brevity. 
	The main requirement is hence to show that the parameter-independent bound $g$ for the profile score $\psi_{\boldsymbol{\theta}}$ of Lemma \ref{e18} is integrable. This dominating condition is due to A. Wald \cite[see e.g.][Sect. 24.2.3, Condition (D3)]{gou2}. Using the marginal distribution $F^T$ of Assumption \ref{A2:SquareInt} and $\log(1-t/G)\le 0$ for $0 \le t \le G$, one has $\mathbb{E}_{\fth_0}[g(X_1,T_1)]
	\leq \int_0^G \int_0^{\infty} |K_{\varepsilon}+K_{\varepsilon}[1-\log(1-t/G)]|  f_{\fth_0}(x,t) \dt x \dt t
	=K_{\varepsilon}+ \frac{K_{\varepsilon}}{G}\int_0^G [1-\log(1-t/G)] \dt t
	=K_{\varepsilon}(1+\frac{1}{G}[2t-G(1-t/G)\log(1-t/G)]_0^G)=3K_{\varepsilon}<\infty$.	\qed    
\end{proof}

Proving normality for the zeros of \eqref{eq:1} for $\fth_0$ in the inner open of $\Theta$ follows \citet[][Theorem 5.41]{vaart1998}. Consistency, given by Theorem \ref{cons}, as well as Lemma \ref{lem1}, are requirements. The remaining arguments are given in the beginning of Appendix \ref{appnorm}.

\begin{theorem}\label{satznormoffen}
	Under Assumptions \textup{\ref{A1:Compact}}-\textup{\ref{A5negdef}}, $\hat{\fth}$ defined after \eqref{eq:1} and $\fth_0 \in (\varepsilon,1/\varepsilon) \times (0,1-\varepsilon_{\vartheta})$ , the sequence  $\sqrt{n}(\hat{\fth}_n-\fth_0)$ converges for $n \to \infty$ in distribution to a normally distributed random variable with expectation (vector) $\mathbf{0}$ and covariance matrix 
	\[
	(\mathbb{E}_{\fth_0}[\dot{\psi}_{\fth_0}(X_i,T_i)])^{-1}\mathcal{I}(\boldsymbol{\theta}_0)(\mathbb{E}_{\fth_0}[\dot{\psi}_{\fth_0}(X_i,T_i)])^{-1}.
	\] 
\end{theorem}

From the theorem a Wald-type test can be performed with the respective confidence interval around $\hat{\vartheta}_n$ for any hypothetical $\vartheta_0>0$. A numerical simplification for such a confidence interval yields the information matrix equality. The second half of Appendix \ref{appnorm} yields under the Assumptions \ref{A1:Compact}-\ref{A4:trunc} the analogue
\begin{equation}\label{ime}
	\mathbb{E}_{\fth_0}[\dot{\psi}_{\fth_0}(X_i,T_i)]=-\mathcal{I}(\boldsymbol{\theta}_0),	
\end{equation}
so that the asymptotic covariance matrix in Theorem \ref{satznormoffen} reduces to  $\mathcal{I}(\boldsymbol{\theta}_0)^{-1}$. 

A simple test to reject the independence hypothesis $H_0: \vartheta_0=0$, which is our main interest here, is a too large $\hat{\vartheta}_n$. For the critical value, a distributional statement for $\vartheta_0=0$, i.e. a parameter $\fth_0$ at the boundary of the parameter space $\Theta^H$, is necessary. 
\begin{theorem}\label{satznormrand}
	Under assumptions \textup{\ref{A1:Compact}}-\textup{\ref{A5negdef}}, $\hat{\fth}_n$ defined after \eqref{eq:1} for $\fth_0 \in \Theta^H=(\varepsilon,1/\varepsilon) \times \{0\}$, the distribution of the sequence  $\sqrt{n}(\hat{\fth}_n-\fth_0)$ converges for $n \to \infty$ towards the mixture of distributions
	\[
	\Phi_{\theta_0}(\mathbf{a}) = \frac{1}{2}F_1^{\theta_0}(\mathbf{a}) +  \frac{1}{2}F_2^{\theta_0}(\mathbf{a})
	\]	
	with $\mathbf{a}=(a_{\theta},a_{\vartheta})'$, where $F_1^{\theta_0}$ is a two-dimensional distribution defined in $-\infty < a_{\theta} < \infty$ and  $a_{\vartheta}>0$ and having in this region the density equal to twice the density of $N_2(\mathbf{0}, \mathcal{I}(\boldsymbol{\theta}_0)^{-1})$. Furthermore $F_2^{\theta_0}$ is a one-dimensional distribution of $\sigma^{(2)}  \tilde{Y}_1$,	 
	concentrated on  $-\infty < a_{\theta} < \infty$ and  $a_{\vartheta}=0$.  The distribution of $\tilde{Y}_1$ is the distribution of $Y_1$, in $(Y_1,Y_2)'$ distributed as $N_2(\mathbf{0}, \mathcal{I}(\boldsymbol{\theta}_0))$, conditional on the inequality $Y_2 +  \mathcal{I}(\boldsymbol{\theta}_0)_{12} \sigma^{(2)} Y_1  \le 0$,
	with $\sigma^{(2)}:=(\mathcal{I}(\boldsymbol{\theta}_0)_{11})^{-1}$. 
\end{theorem}

While the original work of A. Wald uses a linear Taylor expansion of the score and excludes the boundary of the parameter space, additional arguments given in \citet[][Theorem 1]{Moran1971441} allow to include the here important boundary. For our rather simple model, a quadratic Taylor expansion is readily available and \citet[][Theorem 5.41]{vaart1998} becomes applicable with cases $\hat{\vth}_n>0$ and $\hat{\vth}_n=0$. 

\begin{proof}
	By Taylor's theorem there exists a $\tilde{\fth}_n \in \{\fth_0 + t (\hat{\fth}_n-\fth_0) | t \in (0,1)\}$, such that - stacked to two coordinates of $\Psi_n$ with potentially different $\tilde{\fth}_n$ - 
	\begin{equation}\label{e20}
		\Psi_n(\hat{\fth}_n)=\Psi_n(\fth_0)+\dot{\Psi}_n(\fth_0)(\hat{\fth}_n-\fth_0)+\frac{1}{2}(\hat{\fth}_n-\fth_0)^{T}\ddot{\Psi}_n(\tilde{\fth}_n)(\hat{\fth}_n-\fth_0).
	\end{equation}
	Now $\Psi_{n,1}(\hat{\fth}_n)=0$ in any case, but $\Psi_{n,2}(\hat{\fth}_n)= 0$ if $\hat{\vth}_n>0$ and $\Psi_{n,2}(\hat{\fth}_n)\leq 0$ if $\hat{\vth}_n=0$.
	The term $\Psi_n(\fth_0)$ is the average of independent identically distributed random vectors $\psi_{\fth_0}(X_i, T_i)$ with $\mathbb{E}_{\fth_0}[\psi_{\fth_0}(X_i, T_i)]=0$. According to the central limit theorem the sequence $\sqrt{n}\Psi_n(\fth_0)$ converges in distribution towards $N_2(\mathbf{0}, \mathcal{I}(\fth_0))$. The first derivatives $\dot{\Psi}_n(\fth_0)$ in the second term converge (LLN) in probability towards  $\mathbb{E}_{\fth_0}[\dot{\psi}_{\fth_0}(X_1, T_1)]$. The second derivatives $\ddot{\Psi}_n(\tilde{\fth}_n)$ are a two-dimensional vector consisting of $(2 \times 2)$-matrices. According to assumption, there exists a $\delta>0$, such that $\ddot{\psi}_{\fth}(x,t)$ for all $\fth \in (\theta_0-\delta,\theta_0+\delta ) \times [\vth_0, \vth_0+\delta)$ is dominated by an integrable function $\ddot{\psi}(x,t)$. Due to consistency by Theorem \ref{cons}, the probability of $\{\hat{\fth}_n \in (\theta_0-\delta,\theta_0+\delta ) \times [\vth_0, \vth_0+\delta)\}$ converges towards one. On this set now holds
	\begin{equation*}
		\Vert \ddot{\Psi}_n(\tilde{\fth}_n) \Vert = \left\Vert \frac{1}{n} \sum_{i=1}^n \ddot{\psi}_{\tilde{\fth}_n}(X_i,T_i)\right\Vert \leq \frac{1}{n}\sum_{i=1}^n \Vert \ddot{\psi}(X_i,T_i) \Vert.
	\end{equation*}
	Finally due to the integrability of the function $\ddot{\psi}(x,t)$ and the LLN it is bounded in probability. The second and third term of \eqref{e20} can be written as
	\begin{multline*}
		[\mathbb{E}_{\fth_0}[\dot{\psi}_{\fth_0}(X_1, T_1)]+o_P(1)+\frac{1}{2}(\hat{\fth}_n-\fth_0)^{T}O_P(1)](\hat{\fth}_n-\fth_0) \\=	[ \mathbb{E}_{\fth_0}[\dot{\psi}_{\fth_0}(X_1, T_1)]+o_P(1)](\hat{\fth}_n-\fth_0).
	\end{multline*}
	In this case $o_P(1)$ and $O_P(1)$ are $(2 \times 2)$-matrices and a two-dimensional vector, respectively, of  $(2 \times 2)$-matrices, whose entries are sequences that converge towards zero or are bounded, respectively. The right side of the equality follows due to the consistency of Theorem \ref{cons}, which also contains the boundary, so that $(\hat{\fth}_n-\fth_0)'O_P(1)=o_P(1)'O_P(1)=o_P(1)$. By Assumption \ref{A5negdef}, the matrix $\mathbb{E}_{\fth_0}[\dot{\psi}_{\fth_0}(X_1, T_1)]+o_P(1)$ will be invertible for $n$ large enough, explicitly including the case $\vartheta_0=0$. 
	For the case $\hat{\vth}_n>0$ we have 
	\begin{align*}
		0&=\Psi_n(\fth_0)+[ \mathbb{E}_{\fth_0}[\dot{\psi}_{\fth_0}(X_1, T_1)]+o_P(1)](\hat{\fth}_n-\fth_0) \qquad &\Leftrightarrow\\
		\hat{\fth}_n-\fth_0&=-[ \mathbb{E}_{\fth_0}[\dot{\psi}_{\fth_0}(X_1, T_1)]+o_P(1)]^{-1}\Psi_n(\fth_0)  &\Leftrightarrow\\
		\sqrt{n} (\hat{\fth}_n-\fth_0)&=-( \mathbb{E}_{\fth_0}[\dot{\psi}_{\fth_0}(X_1, T_1)])^{-1}\frac{1}{\sqrt{n}} \sum_{i=1}^n \psi_{\fth_0}(X_i,T_i)+o_P(1),
	\end{align*}
	such that $\sqrt{n} (\hat{\fth}_n-\fth_0)$, conditional on $\hat{\vth}_n>0$, towards a two-dimensional distribution  $F_1^{\fth_0}(a_{\theta}, a_{\vth})$ whose density is zero for $a_{\vth}\leq0$ and two times the density function of the normal distribution $N_2(\mathbf{0}, \mathcal{I}(\fth_0)^{-1})$ for $a_{\vth}>0$. For the case $\hat{\vth}_n=0$ one has 
	\begin{align*}
		\Psi_n(\hat{\fth}_n)&= \Psi_n(\fth_0)+[\mathbb{E}_{\fth_0}[\dot{\psi_{\fth}}](X_1,T_1)+o_P(1)](\hat{\fth}_n-\fth_0)\\
		&=\left(\begin{array}{c} \Psi_{n,1}(\fth_0) \\ \Psi_{n,2}(\fth_0) \end{array}\right)-\left( \begin{array}{cc} \mathcal{I}(\fth_0)_{11} & \mathcal{I}(\fth_0)_{12} \\ \mathcal{I}(\fth_0)_{21} & \mathcal{I}(\fth_0)_{22}\end{array}\right) \left( \begin{array}{c} \hat{\theta}_n-\theta_0 \\ 0 \end{array}\right)+ \left( \begin{array}{c} o_P(1) \\ o_P(1) \end{array}\right)\\
		&=\left(\begin{array}{c} \Psi_{n,1}(\fth_0) \\ \Psi_{n,2}(\fth_0) \end{array}\right)-\left(\begin{array}{c} \mathcal{I}(\fth_0)_{11} (\hat{\theta}_n-\theta_0) \\ \mathcal{I}(\fth_0)_{21} (\hat{\theta}_n-\theta_0) \end{array}\right)+ \left( \begin{array}{c} o_P(1) \\ o_P(1) \end{array}\right).
	\end{align*}
	With $\Psi_{n,1}(\hat{\fth}_n)=0$ und $\Psi_{n,2}(\hat{\fth}_n)\leq 0$ hence, it follows
	\begin{align}\label{e21}
		0&= \Psi_{n,1}(\fth_0)-\mathcal{I}(\fth_0)_{11} (\hat{\theta}_n-\theta_0)+o_P(1) \qquad \Leftrightarrow\nonumber\\
		\hat{\theta}_n-\theta_0&=(\mathcal{I}(\fth_0)_{11})^{-1}\Psi_{n,1}(\fth_0)+o_P(1)
	\end{align}
	and
	$0\geq \Psi_{n,2}(\fth_0)-\mathcal{I}(\fth_0)_{21} (\hat{\theta}_n-\theta_0)+o_P(1)$.
	Furthermore, by insertion of equation \eqref{e21} it follows
	$0\geq\Psi_{n,2}(\fth_0)-\mathcal{I}(\fth_0)_{21}(\mathcal{I}(\fth_0)_{11})^{-1}\Psi_{n,1}(\fth_0) +o_P(1)$.
	Hence, under the condition $0\geq Y_2-\mathcal{I}(\fth_0)_{21}(\mathcal{I}(\fth_0)_{11})^{-1}Y_1$,  $\sqrt{n}(\hat{\theta}_n-\theta_0)$ is  asymptotically normal with expectation  $0$ and variance $(\mathcal{I}(\fth_0)_{11})^{-1}$, where $(Y_1,Y_2) \sim N_2(\mathbf{0}, \mathcal{I}(\fth_0))$. \qed
\end{proof}

Note that under the independence hypothesis of the test $H_0: \vartheta_0=0$, by Theorem \ref{satznormrand}, the probability of a boundary value $\hat{\boldsymbol{\theta}}_n$, i.e. with $\hat{\vartheta}_n=0$, has probability 0.5, as announced in Section \ref{secident}.  

For performing the test against $H_0$ on basis of a too large $\hat{\vartheta}$ using Theorem \ref{satznormrand}, the estimation of $\mathcal{I}(\fth_0)$ under $H_0$ is needed. Denote now by $\hat{\boldsymbol{\theta}}^0$ the restricted estimator, under $H_0$, namely $\hat{\boldsymbol{\theta}}^0=(\hat{\theta}^0, 0)$. Now $\mathcal{I}(\hat{\boldsymbol{\theta}}^0)$ can be replaced by a consistent estimate due to Slusky's Lemma \cite[see e.g.][Lemma 2.8]{vaart1998}, e.g. due to the LLN by 
\[
\frac{1}{n}\sum_{i=1}^n \psi_{\hat{\boldsymbol{\theta}}^0}(X_i,T_i) \psi_{\hat{\boldsymbol{\theta}}^0}(X_i,T_i)'=\frac{1}{n}\sum_{j=1}^{M} \psi_{\hat{\boldsymbol{\theta}}^0}(\tilde{X}_j,\tilde{T}_j)\psi_{\hat{\boldsymbol{\theta}}^0}(\tilde{X}_j,\tilde{T}_j)'
\] \cite[see e.g.][Remark 17.4]{gou2}. Note that the three $\boldsymbol{\theta}$'s in the definition of $\mathcal{I}(\boldsymbol{\theta})$ in \eqref{deffi} are treated differently. The first, in $\mathbb{E}_{\boldsymbol{\theta}}$,
is implicitly replaced by $\boldsymbol{\theta}_0$, being $(\theta_0,0)'$ under $H_0$, by averaging with respect to the respective distribution. Those in $\psi_{\fth}$ are explicitly replaced by $\hat{\boldsymbol{\theta}}^0$.

\subsection{Empirical example} \label{empexa}

We reconsider, from \citet{weiswied2021}, $m=55{,}279$ enterprise lifetimes $\tilde{x}_j$ ending in the observation period 2013-2015 (see Figure \ref{exa0}). Those had been at risk of closure for  $\sum_{j=1}^m \tilde x_j=0{.}54$ million years. Also, for each enterprise, the date of foundation, and hence the age at the beginning of 2013, $\tilde{t}_j$, is known. For Assumptions \ref{A1:Compact}, \ref{A2:SquareInt} and \ref{A4:trunc}, without dependence, $\hat{\theta}^0 = 0.08261$ had been reported.

\subsubsection{Testing for a positive trend in life expectancy} \label{empexa1}

The estimate defined after \eqref{eq:1} for the model given by Assumptions \ref{A1:Compact}-\ref{A5negdef} is the minimum of the negative profile likelihood depicted in Figure \ref{appl_power} (left). Irrespective of $\theta$, the likelihood decreases as function of $\vartheta$, even though slowly, so that $\hat{\vartheta}_n=0$. Similar to but, not equal to the restriction of independence,  $\hat{\theta}_n = 0.08257$.  Under the hypothesis, $H_0: \vartheta_0=0$, such boundary value, not indicating a positive trend in life expectancy, has by Theorem \ref{satznormrand} a probability 0.5, being also the $p$-value of the test.

\begin{figure}[]
	\centering
	\begin{minipage}{0.45\linewidth}
		\includegraphics[width=6cm]{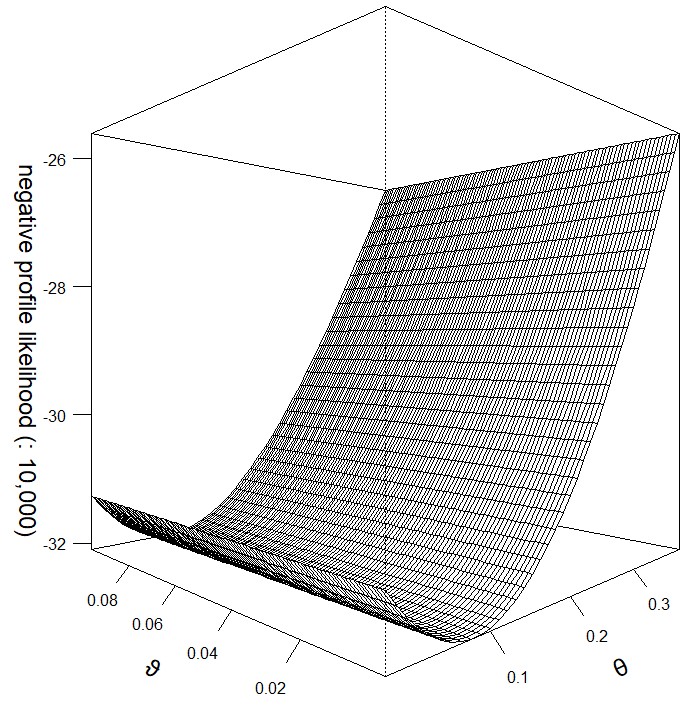} 
	\end{minipage}
	\hspace{0.3cm}
	\begin{minipage}{0.48\linewidth}
		\includegraphics[width=6.8cm]{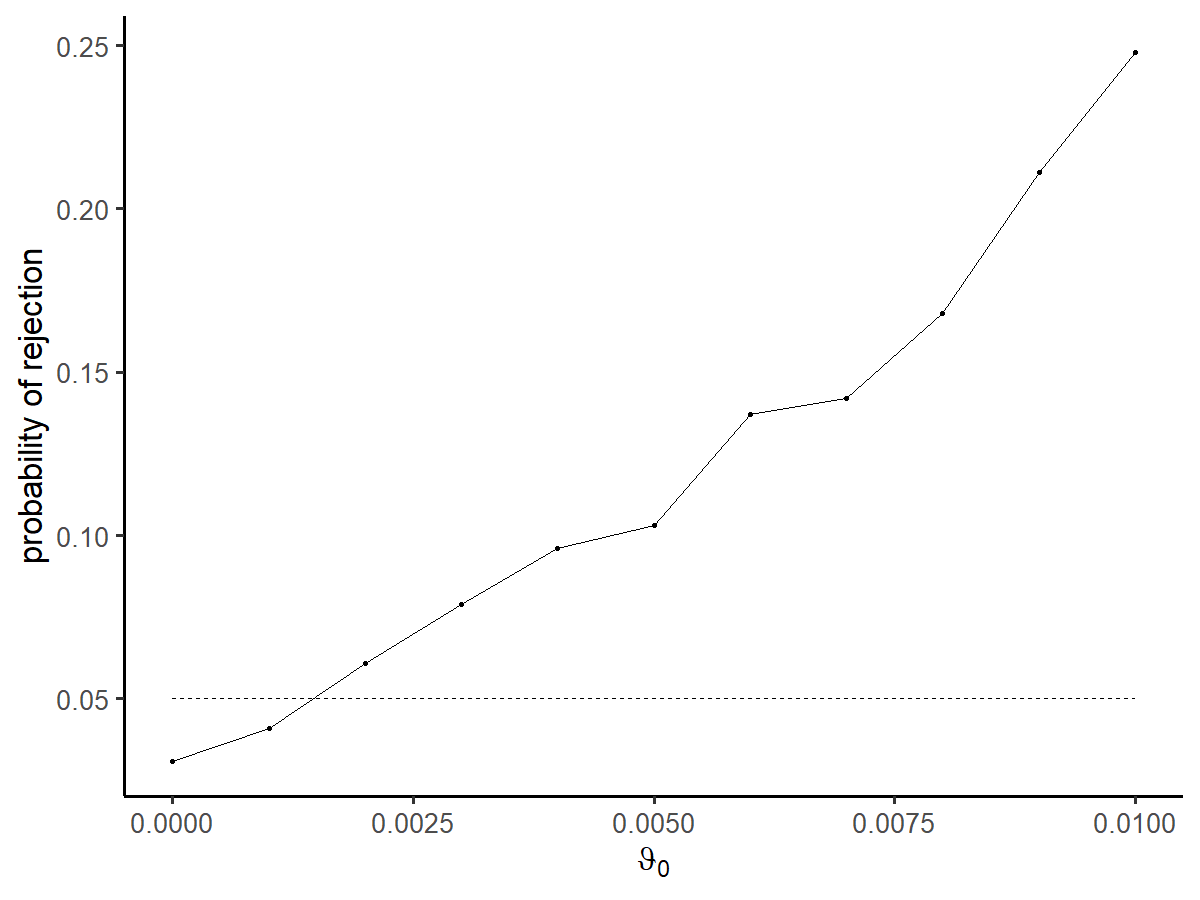} 
	\end{minipage}
	\caption{Left: Profile (logarithmic) likelihood for example /Right: Simulated actual level ($\vartheta_0=0$) and power ($\vartheta_0>0$) of test for independence given by Theorem \ref{satznormrand} (under conditions of example in Section \ref{empexa} using $\theta_0=0.08$, $G=24$ and $s=3$)}
	\label{appl_power}
	(Explanation of panels is distributed over larger parts of text.)
\end{figure}

\subsubsection{Testing for a negative trend in life expectancy}

Given the result in Section \ref{empexa1}, even a negative trend in German business life expectancy seems conceivable and we model it with the Farlie-Gumbel-Morgenstern copula introduced in the Assumption (FGM) (replacing Assumption (A3)). Note that, in contrast to a ordinary linear regression, it is not be possible to calculate a negative trend with the Gumbel-Barnett copula by transforming the dependent variable.
Again as in \eqref {e6} we estimate the parameter $\fth$ with the first two coordinates of 
\begin{align*}
	\argmax\limits_{\fth \in \Theta, n \in \mathbb{N}} \,\log \ell_{FGM} (\fth,n)
\end{align*}
after profiling out $n$. Here $\ell_{FGM}$ replaces $\ell$ (of \eqref{e5}) with $f_{\fth}$  instead of $f^{FGM}_{\fth}$, and  $\alpha_{\boldsymbol{\theta}}$ by $\alpha_{\boldsymbol{\theta}}^{FGM}$. It results in $\hat{\theta}^{FGM}_n=0{.}0817$, i.e. the dependence model suggests a similar average lifetime of $\approx 12{.}24$ years as in \citet{weiswied2021}, without dependence. What is new is that $\hat{\vth}^{FGM}_n=0{.}10$, which corresponds to a Kendall's $\tau$ of $0{.}03$ and a negative time trend. With
$\mathbb{E}[X|T=t]=\int_0^{\infty} x f_{\fth}^{FGM}(x,t)/f^T(t) dx =\theta^{-1}[1-0{.}5 \vth^{FGM} (1-2t/G)]$, an annual decrease in the life expectancy of $\hat{\vth}^{FGM}_n/(\theta G)$, i.e. of around 19 days, results.

\section{Behaviour in finite samples} \label{mcs}
 
We conduct a Monte Carlo simulation, primarily to visualize the asymptotic results on consistency given by Theorem \ref{cons}, measured in mean squared error (MSE), decomposed into bias and variance. In particular, we will find that the asymptotic approximation is rather precise regarding our statements on basis of business closure data in Section \ref{empexa}. Also the actual level and power of the test, given by Theorem \ref{satznormrand}, are studied.

\subsection{Algorithm for simulating truncated sample} \label{algosim}

In order to generate the latent sample of $n$ measurements $(X_i,T_i)'$ of Section \ref{secpop2}, consider the conditional inversion method using a copula $C$ \cite[see][Section 2.9]{B3}. Note first that the inverse of $c_u^{\vartheta}$, introduced shortly after \eqref{alpha0}, exists. 

\begin{algo}\label{algo1} Generation of $(X_i,T_i)'$ under Assumptions \ref{A1:Compact}-\ref{A3:Ind}: 
	\begin{itemize}[itemsep=0pt]
		\item[(i)] Generate independent realisations $U$ and $\check{V}$ from $Unif(0,1)$. 
		\item[(ii)] Set $V=(c_u^{\vartheta})^{-1}(\check{V})$.
		\item[(iii)] Set $X_i=F_E^{-1}(U)$ and $T_i=(F^T)^{-1}(V)$.
	\end{itemize}
\end{algo}
The observations $(\tilde x_1, \tilde t_1)' \ldots, (\tilde x_m, \tilde t_m)'$ of Section \ref{secpop3} then arise by imposing Assumption \ref{A4:trunc}, i.e. by truncating $(X_i,T_i)' \not\in D$.

\subsection{Choices for parameters and sample size}

The simulation extends the case of independent truncation, that is, $\vartheta_0=0$ in \cite{weiswied2021}. 
The considered sample sizes $n \in \{10^p, p=3, \ldots, 5 \}$, the widths of the population $G \in \{24,48\}$ and of the observation period $s \in \{2,3,48\}$ are the same as in \citet[][Section 4]{weiswied2021}. For the exponential parameter of Assumption \ref{A2:SquareInt}, we choose $\theta_0 \in \{0.1, 0.05\}$, being an excerpt of \citet[][Section 4]{weiswied2021}. We now newly include values $\vartheta_0 \in \{0.001,0.01\}$ for copula dependence of Assumption \ref{A3:Ind}, leading to Kendall's $\tau \in \{-0.0005, - 0.005 \}$. By doing so, we model weak dependence.

\subsection{Result} \label{simressec}

Each scenario consists of $G$, $s$, $n$, $\theta_0$ and $\vartheta_0$. A first impression of the asymptotic fit can be gained for the test on independence given by Theorem \ref{satznormrand}. Figure \ref{appl_power} (right) depicts simulated rejection rates of simulated datasets (as of Section \ref{algosim}). The rate is the actual level of the test at nominal level of 5\% for $\vartheta_0=0$ and it approximates the power for $\vartheta_0>0$. It can be seen that the test is slightly conservative, as the actual is below 5\% at the origin, but quickly exceeds the nominal level and has power of 25\% already at a value as small as $\vartheta_0=0.01$.  

We now study the bias and variance of $\hat{\fth}_n$. The finite sample biases of the estimators $\hat{\theta}$ and  $\hat{\vartheta}$ as zeros of the system of equations \eqref{eq:1} (here and in Table \ref{sims} omitting the subscript $n$) are approximated from the $R=1000$ simulated data sets, due to Algorithm \ref{algo1} by $\frac{1}{R}\sum_{\nu=1}^R \hat{\theta}^{(\nu)} - \theta_0$ and $\frac{1}{R}\sum_{\nu=1}^R \hat{\vartheta}^{(\nu)} - \vartheta_0$.  Table \ref{sims} in Appendix \ref{ressec} lists the results, and it can be seen that the bias of $\hat{\theta}$ decreases to virtually zero as a function of $n$ for all scenarios, all $n$ are smaller than in our example of Section \ref{empexa}. The bias of $\hat{\vartheta}$ is markedly larger than of $\hat{\theta}$ in general, but also decreases in $n$. 

In order to conclude consistency in probability, consider the MSE as the sum of squared bias and variance $Var(\hat{\theta})$ (alike for $\vartheta$). 
The simulated approximation of the variance is $\frac{1}{R}\sum_{\nu=1}^R (\hat{\theta}^{(\nu)} - \theta_0)^2$ (alike for $\hat{\vartheta}$). Evident from Table \ref{sims} is the generally small variance of $\hat{\theta}$, quickly decreasing in $n$, and the quite large and also decreasing variance of $\hat{\vartheta}$. Hence the MSE's are approaching zero and consistency is visible for realistic sample sizes. In that respect, note that for $\vartheta_0=0$, the number of observations $m$ is around 1\%-8\% of the sample size $n$  \cite[see][Table 1]{weiswied2021} for chosen $\theta_0$'s and similarly for $\vartheta_0 \in \{0.001, 0.01\}$.

The influence of $G$ and $s$, is, as expected, that for large $s$, i.e. by observing more units, the insecurity about the parameter, i.e. its variance, decreases. For instance, for the scenarios with $G=24$, combined with $s=2$, $s=3$ or $s=48$, exhibit the tendency that  $Var(\hat{\theta})$ is much smaller for $s=48$. The effect of different $G$ is mixed. Simulations not shown here exhibit that the estimation variances first increase, as a function of $G$, then decrease and later increase again. Normality for small sample sizes, as indicated by Theorem \ref{satznormoffen}, will be valid as in the case of independent left-truncation \cite[see][Appendix A.1.4]{weissbachm2021effect} and is not studied in detail. 

As direct comparison with independence \citep{weiswied2021}. It is mainly evident that the dependence introduces a (higher) bias in the estimation of $\theta_0$. For instance, in the scenario with $G=24,s=3$ and $\theta_0=0.05$, the bias is roughly ten times higher for all $n$.

\section{Discussion} \label{secdiscussion}

In some sense, our study follows up on \cite{efron1999} who assumed independence between units when sampling from the latent population and at the same time also for sampling from the truncated population. We describe how the two assumptions can be both true.   

In view of applicability of the results, the presented model is, at least, larger than the model without dependence in \cite{weiswied2021}. However, the model is still very small, and especially assuming the distribution of the lifetime to be exponential can be inadequate, e.g. in human demography. At the other hand, a nonparametric model generates considerable algorithmic effort \cite[see e.g.][]{efron1999,shen2010} and the resulting functional estimator still does not allow statements about popular characteristics such as the expectation or any quantile. Against our assumption of a uniform truncation distribution is a recent finding of \cite{weisdoer2022}, who showed that business foundations became less and less frequent over the years 1990-2013, however assuming independent truncation. And of course the dependence model could be inadequate and \cite{chiou2019} call for a goodness-of-fit test. Furthermore a covariate can be available and informative, it may, for instance, reduce or substitute the dependence. An important example is that truncation dependence can be interpreted as dependence of the lifetime on the date of birth as a cohort effect. And comparisons with methods which incorporate calendar time as a time-dependent covariate found in \cite{renxi2019,franchae2019} should be interesting. For any covariate, such as place of business, it should also be noted that, together with the parameter which relates the covariate to the hazard rate, the marginal distribution of the covariate introduces parameters that can be estimated, or conditioning must be studied in order to avoid a joint estimation \cite[e.g. as in][]{weisdoer2022}.  

In the view of theory, inference about the unknown sample size $n$ would be interesting, if relevant applications can be found. 

In the view of our particular measurement, the elementary biological question of the time between a well-defined birth and a well-defined death, appears to overly simplify business demography. Even when we only wish to study business closure as analogous to human death. In fact, the data in Section \ref{empexa} contain `only' insolvencies, i.e. only closures for one particular cause. Data for a competing risk model would be needed. Richer data would probably then be left-truncated and right-censored (LTRC), rather than DT. Our test can still be applied to LTRC data by dropping the right-censored observations.

\textbf{Conflict of interest}:
 The authors declare not to have conflict of interest.
 

\appendix
\renewcommand{\theequation}{\thesection.\arabic{equation}}
\numberwithin{equation}{section}
\newpage 
\section{Visualization of Gumbel-Barnett copula} \label{copvis}
\begin{figure}[H]
	\centering
	\begin{minipage}{0.49\linewidth}
		\includegraphics[width=5.5cm]{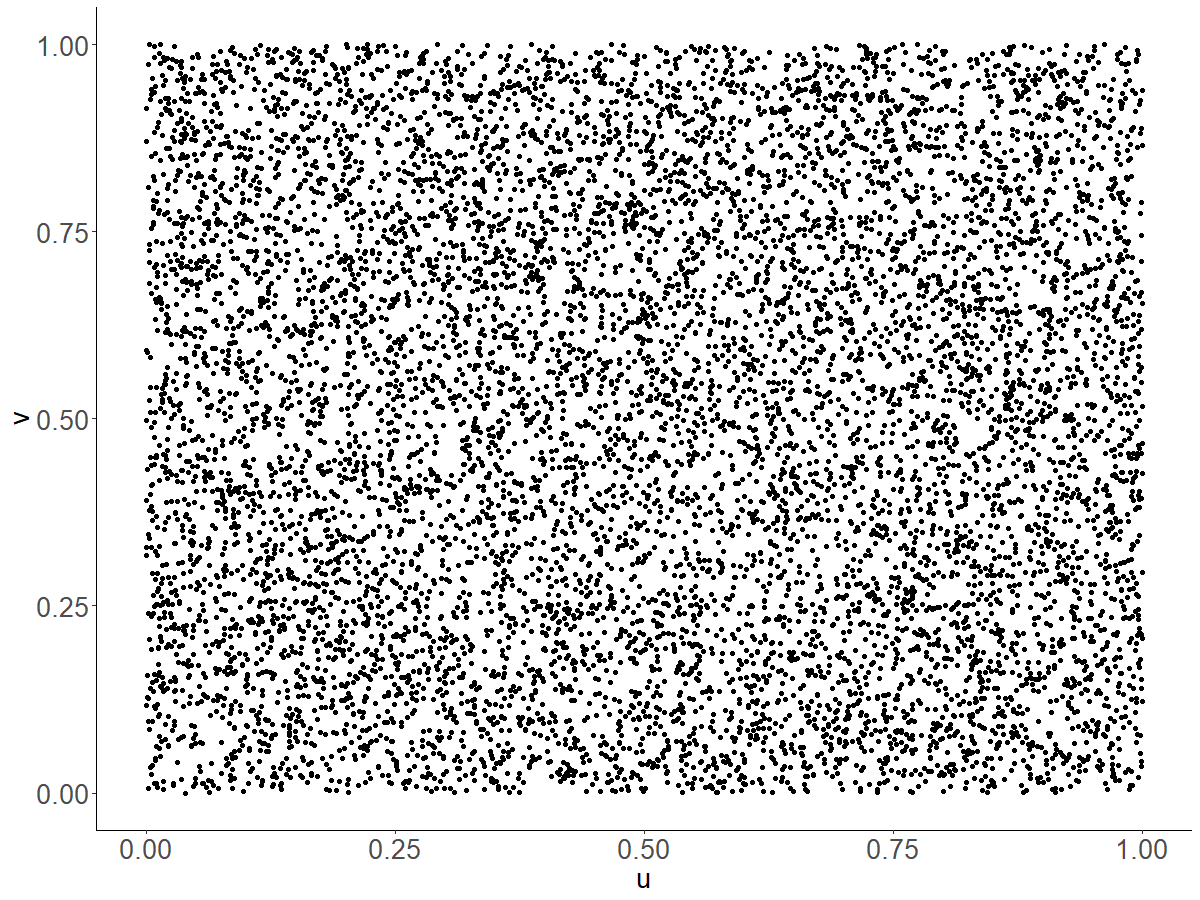}
	\end{minipage}
	\begin{minipage}{0.5\linewidth}
		\includegraphics[width=5.5cm]{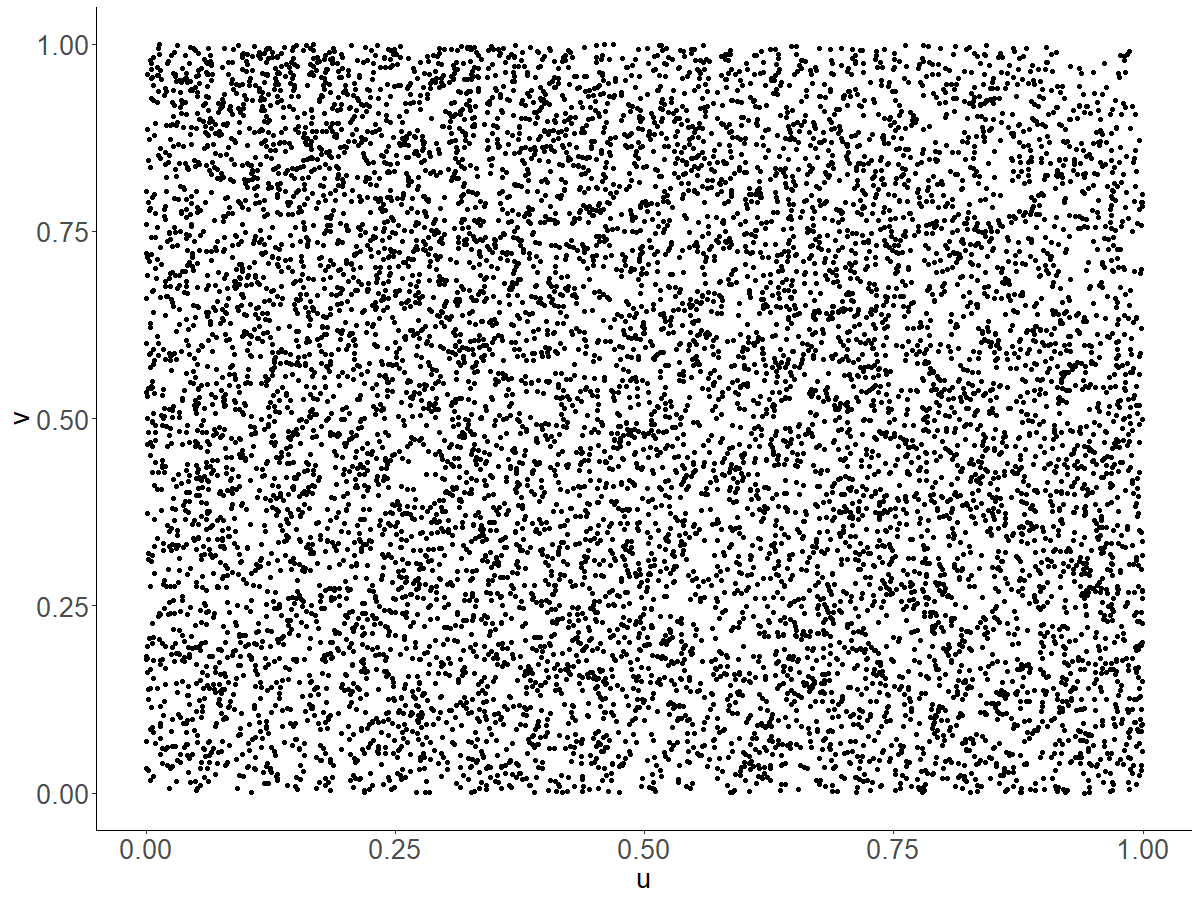}
	\end{minipage}
	
	\begin{minipage}{0.49\linewidth}
		\includegraphics[width=5.5cm]{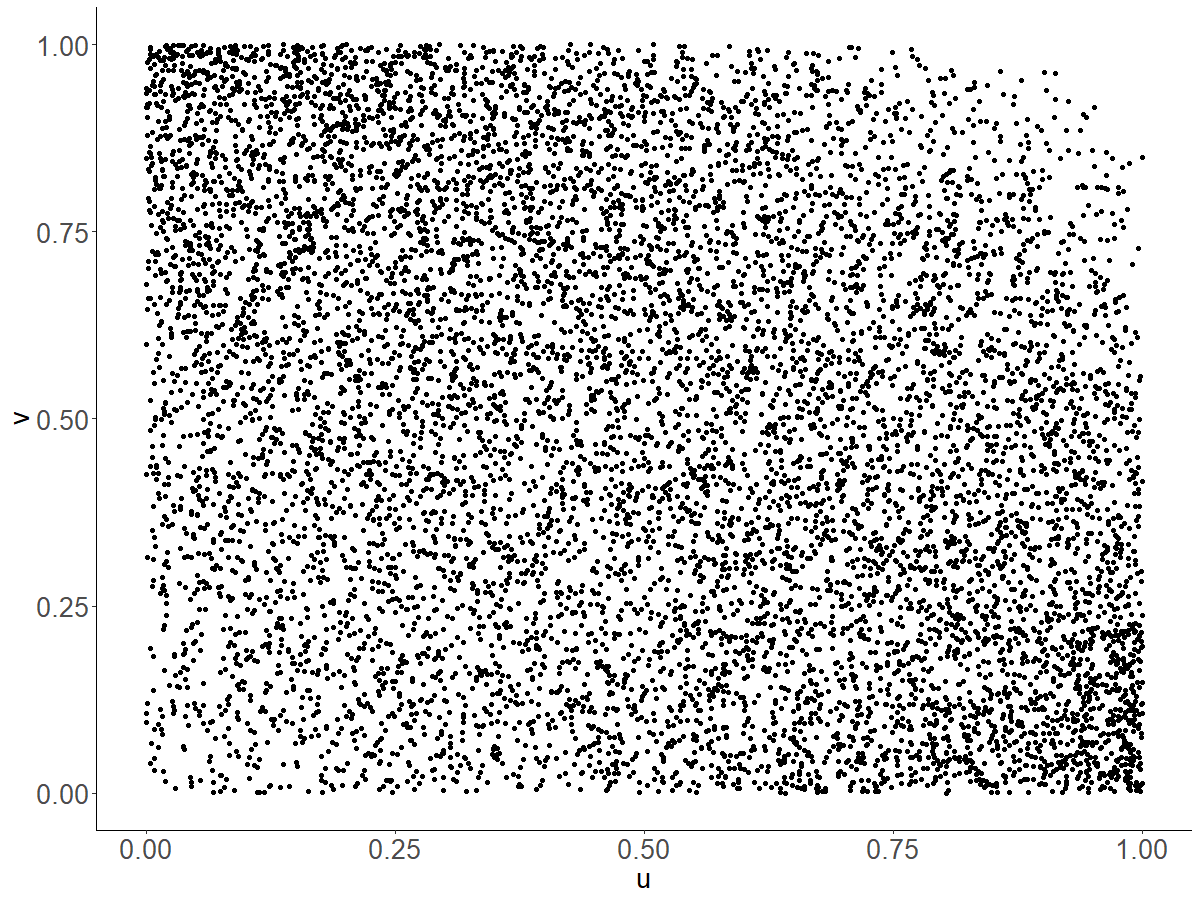}
	\end{minipage}
	\begin{minipage}{0.5\linewidth}
		\includegraphics[width=5.5cm]{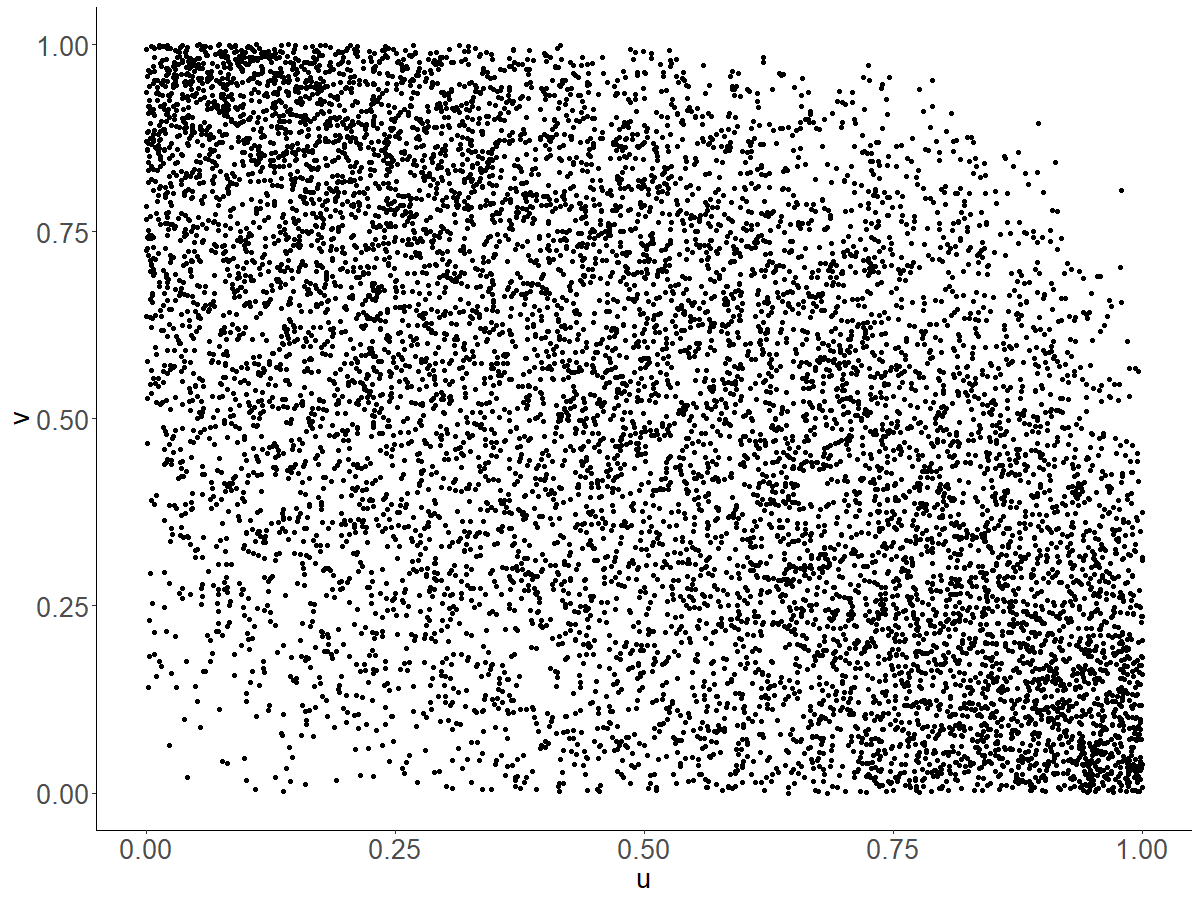}
	\end{minipage}
	
	\begin{minipage}{0.49\linewidth}
		\includegraphics[width=5.5cm]{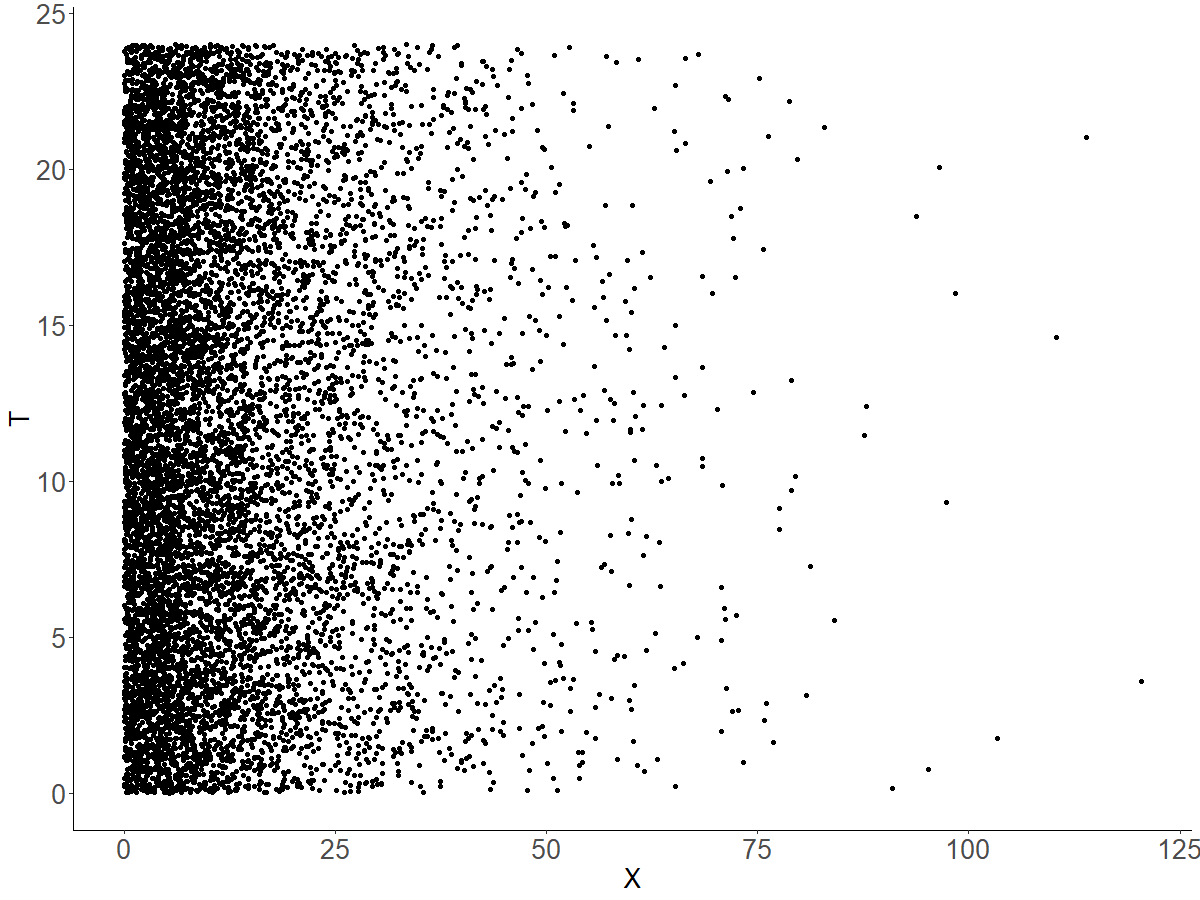}
	\end{minipage}
	\begin{minipage}{0.5\linewidth}
		\includegraphics[width=5.5cm]{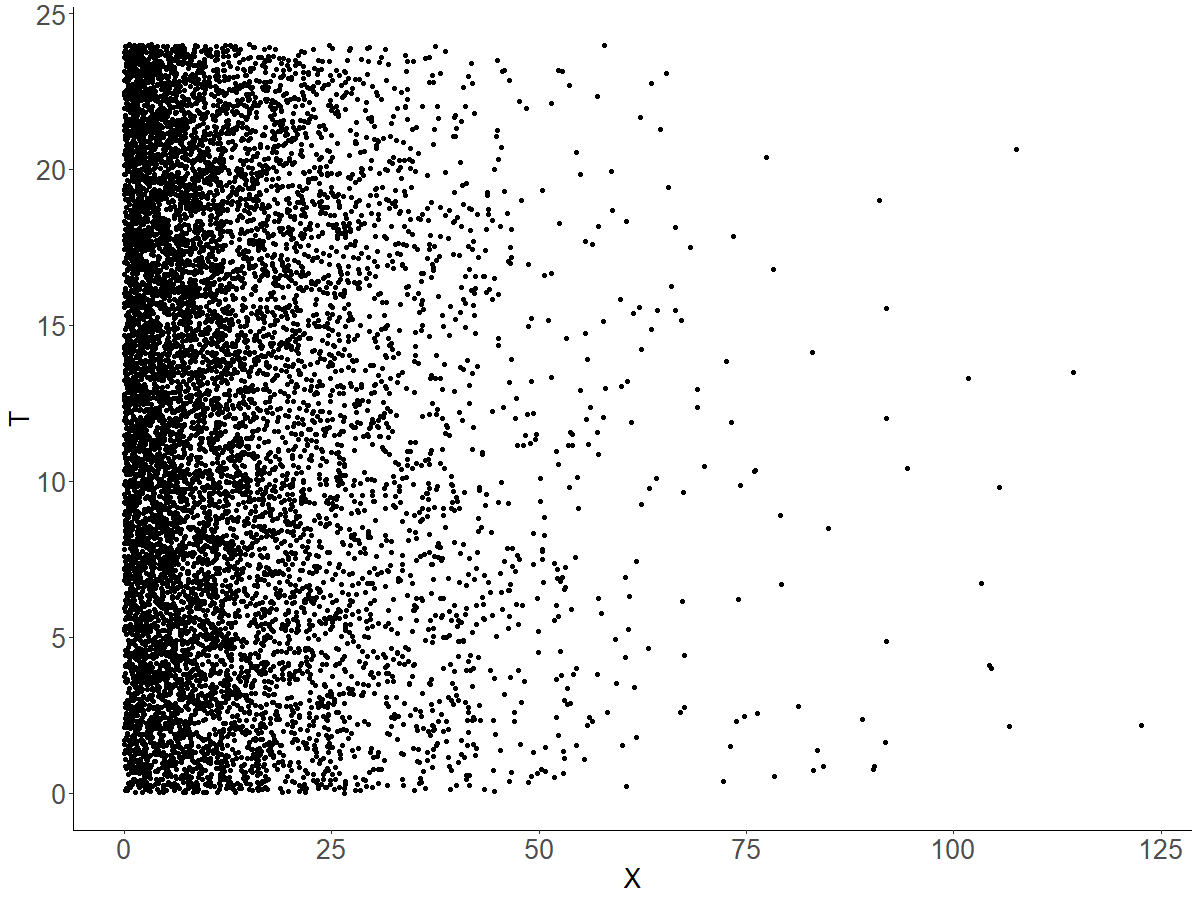}
	\end{minipage}
	
	\begin{minipage}{0.49\linewidth}
		\includegraphics[width=5.5cm]{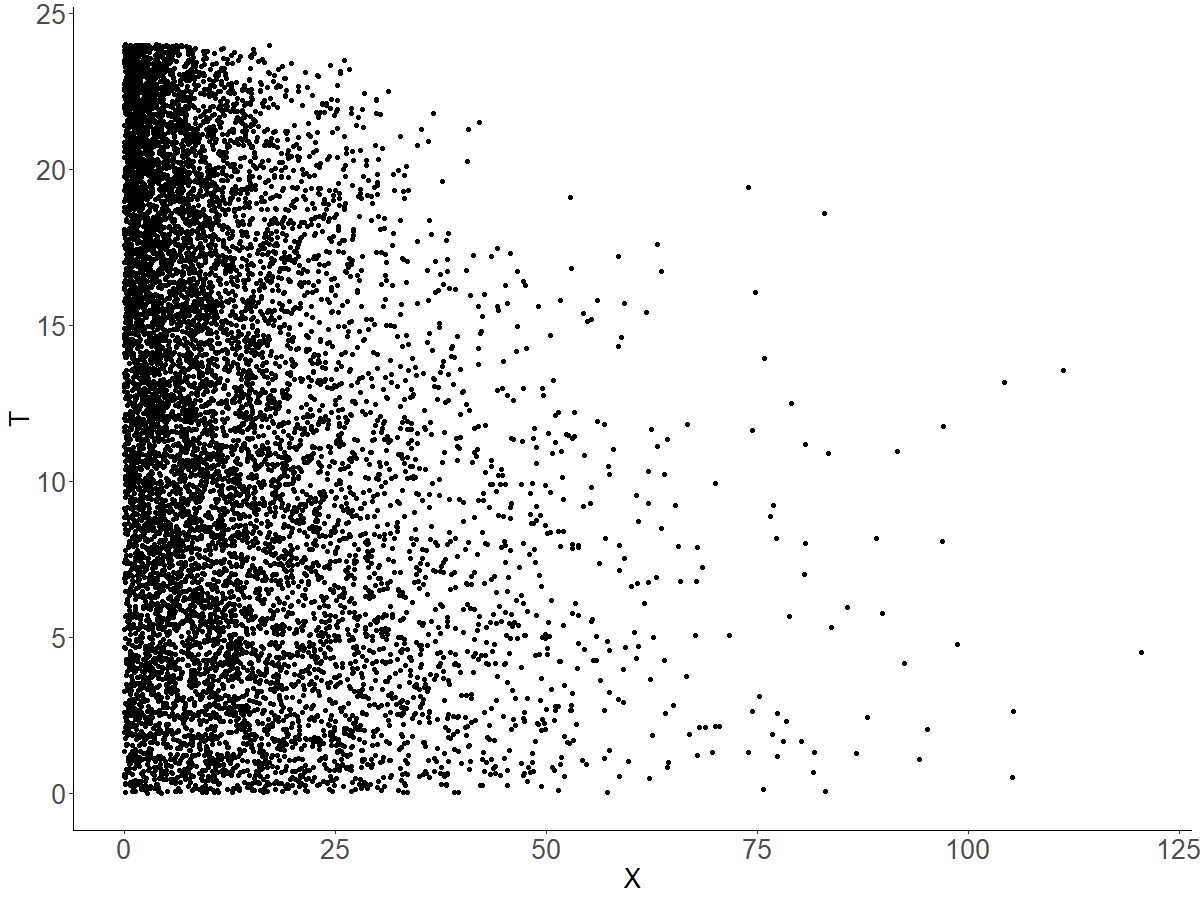}
	\end{minipage}
	\begin{minipage}{0.5\linewidth}
		\includegraphics[width=5.5cm]{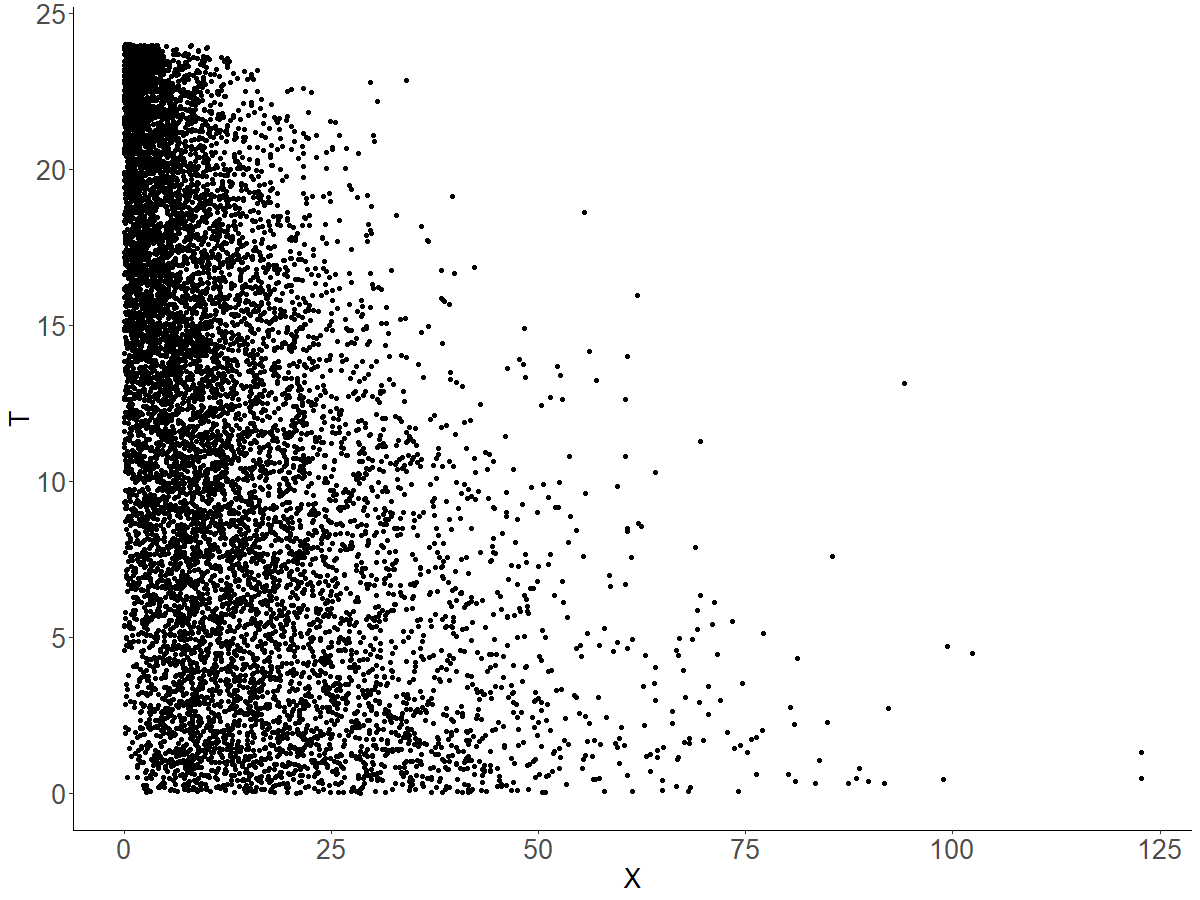}
	\end{minipage}	\caption{Scatter plot of $10,000$ bivariate random draws according to copula of Assumption \ref{A3:Ind} by Algorithm \ref{algo1} with parameter values $\vth=0$, $\vth=0.1$ (rows 1+3), $\vth=0.5$ and $\vth=1- \varepsilon_{\vartheta}$ (rows 2+4, with negligible $\varepsilon_{\vartheta}$): both margins uniform (top)/ one margin exponential (x-axis, $\theta=0.08$, bottom)} \label{punkte1}
\end{figure}

\section{Visualization of FGM copula} \label{simfgm}

\begin{figure}[H]
		\includegraphics[width=4.5cm]{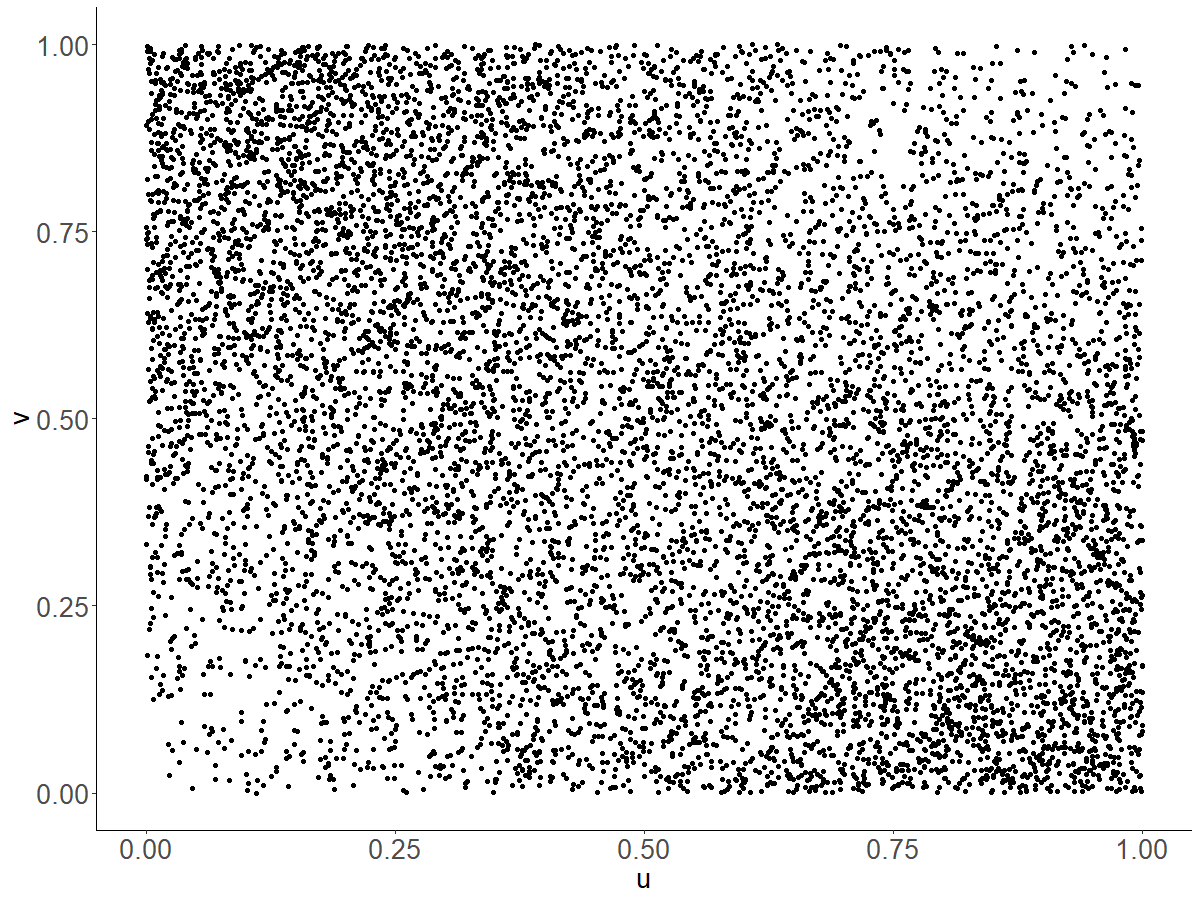}
		\includegraphics[width=4.5cm]{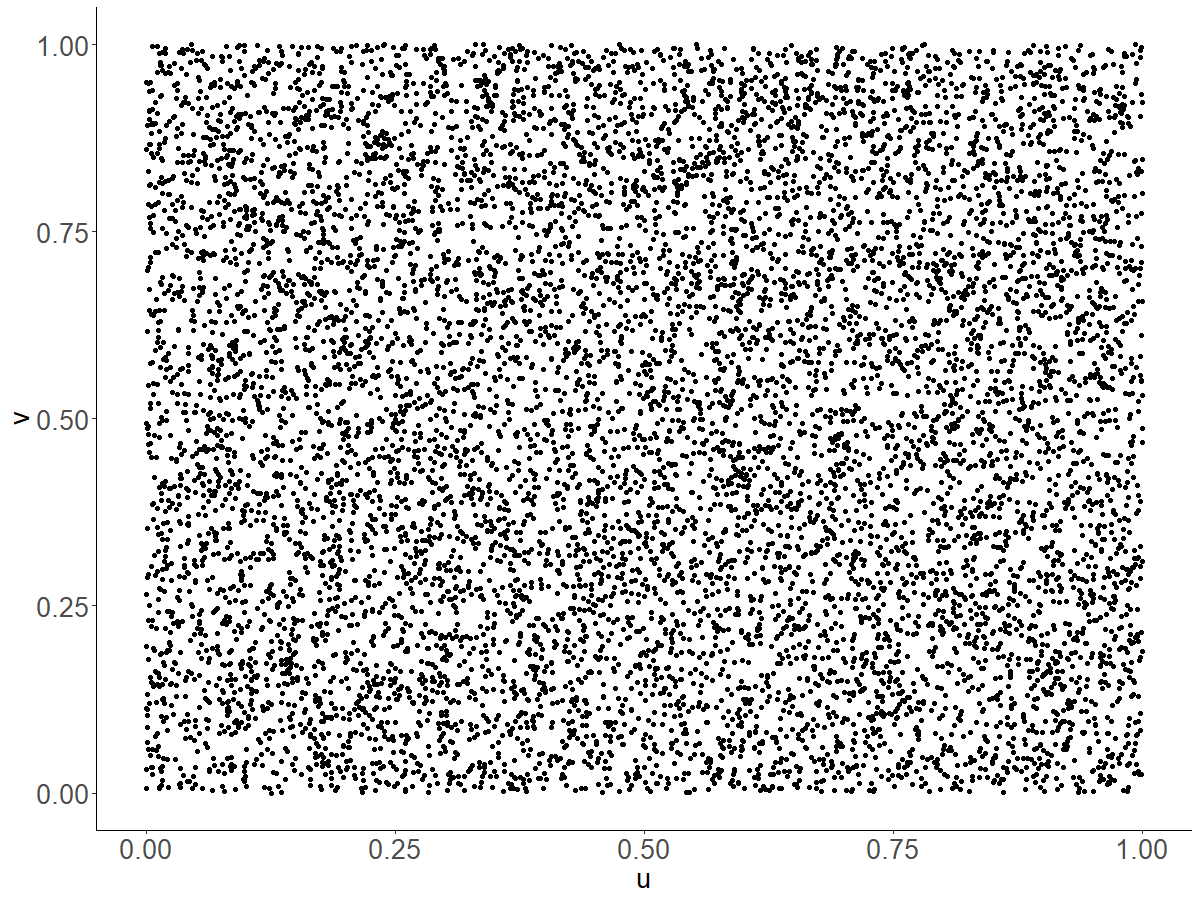}
		\includegraphics[width=4.5cm]{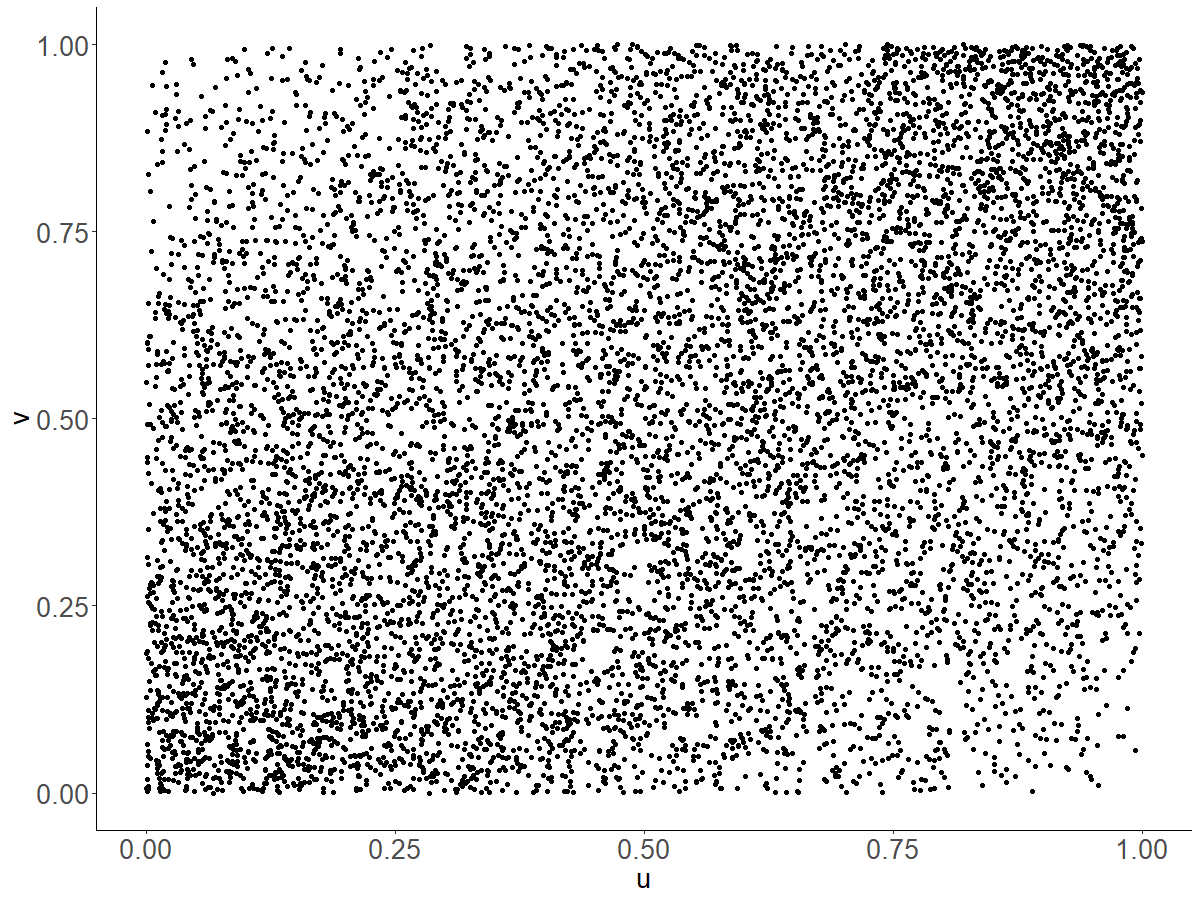}
	\caption{Scatter plot of $10{,}000$ realizations for copula $C^{\vartheta}_{FGM}$ with $\vth^{FGM}=\varepsilon_{\vartheta}-1$ (left), $\vth^{FGM}=0$ (middle) and $\vth^{FGM}=1-\varepsilon_{\vartheta}$ (right), with negligible $\varepsilon_{\vartheta}$}
	\label{abb3}
\end{figure}

\section{Derivations for Section \ref{seclike}}
\subsection{Proof of Lemma \ref{h1}} \label{appendixlem5}

Let $f,\,g: \, (0,\infty) \to \mathbb{R}$ be two functions with $f(x)=\log(x)-(1-x^{-1})$ and  $g(x)=(x-1)-\log(x)$ for $x \in (0,\infty)$. Obviously, it is $f(1)=g(1)=0$. Furthermore, the functions are strictly increasing on the interval $[1,\infty)$, because the derivatives $f'(x)=x^{-1}-x^{-2}$ and $g'(x)=1-x^{-1}$ are positive on the interval $(1,\infty)$ and zero for $x=1$. For all $x>1$ it follows
\begin{align*}
	f(1) < f(x) \Longleftrightarrow 0 < \log(x)-\left(1-\frac{1}{x}\right) \Longleftrightarrow 1-\frac{1}{x} < \log(x),\\
	g(1) < g(x) \Longleftrightarrow 0 < (x-1)-\log(x) \Longleftrightarrow \log(x) < x-1.
\end{align*} \qed

\subsection{Score of $\ell$} \label{appscore}

\subsubsection{Coordinates $\theta$ and $\vartheta$} \label{score_theta_vartheta}

\begin{eqnarray*}
	\frac{\partial}{\partial \theta} \log \ell
	&= & \frac{M}{\theta}+\sum_{j=1}^M \widetilde{X}_j (\vth\log(1-\widetilde{T}_j/G)-1)-n	\frac{\partial\alpha_{\fth}}{\partial \theta} \nonumber \\
	& & +\sum_{j=1}^M \frac{\vth \widetilde{X}_j(\vth \log(1-\widetilde{T}_j/G)-1)}{(\vth \theta \widetilde{X}_j+1)(\vth \log(1-\widetilde{T}_j/G)-1)+\vth} \label{scoretheta}\\
	\frac{\partial}{\partial \vth} \log \ell
	&= & \theta \sum_{j=1}^M \widetilde{X}_j \log(1-\widetilde{T}_j/G)-n 	\frac{\partial\alpha_{\fth}}{\partial \vth} \nonumber  \\
	& & +\sum_{j=1}^M \frac{(2\vth \theta \widetilde{X}_j+1)\log(1-\widetilde{T}_j/G)-\theta \widetilde{X}_j +1}{(\vth \theta \widetilde{X}_j+1)(\vth \log(1-\widetilde{T}_j/G)-1)+\vth}\label{scorevartheta}
\end{eqnarray*}

Partial derivatives of $\alpha_{\fth}$ are not given explicitly but are easy to calculate numerically due to the boundedness of $D$.

\subsubsection{Coordinate $n$} \label{score_n}

It is sufficient to consider the terms that depend on $n$ in log-likelihood \eqref{e5}. We search now the largest value for $n \in \mathbb{N}$, so that	$M  \log(n)-n \alpha_{\fth} \leq M  \log(n+1)-(n+1) \alpha_{\fth}$. By algebraic equivalance and application of  Lemma \ref{h1} we have $\alpha_{\fth}/M\leq \log((n+1)/n)<1/n$,
so that $n$ is the largest number, strictly smaller than $M/\alpha_{\fth}$.

\section{Identification} 

\subsection{Latent population model} \label{identsrs}

Consider an SRS of the joint distribution (see Figure \ref{exa}, top boxes). And keep in mind that with only $\theta$ being of interest, $\vartheta$ would not be a nuisance parameter and the $T_i$'s were abundant. Hence we assume interest in $\boldsymbol{\theta}$ here.   
From $f_{\fth_2}(x,t)/f_{\fth_1}(x,t)=1$, for all ${(x,t) \in D}$, we may conclude equality $\fth_1=\fth_2$. 
To this end, insert definition \eqref{e4}, to see that the density ratio is
\begin{equation} \label{dichtquot}
	\frac{\theta_2}{\theta_1} e^{-x(\theta_2-\theta_1)} \left(1-\frac{t}{G}\right)^{x(\theta_2\vartheta_2-\theta_1\vartheta_1)} \frac{(\vartheta_2\theta_2 x+1)\left(\vartheta_2 \log\left(1-\frac{t}{G}\right)-1\right)+\vartheta_2}{(\vartheta_1\theta_1 x+1)\left(\vartheta_1 \log\left(1-\frac{t}{G}\right)-1\right)+\vartheta_1}=1.
\end{equation}
Taking derivatives in the direction of $t$ and take $T_{\fth}:=(\vartheta\theta x+1)(\vartheta \log(1-t/G)-1)+\vartheta$, results into $F_1 F_2 =0$, with
\begin{eqnarray*}
	F_1& := & -(\theta_2/(\theta_1 G)) e^{-x(\theta_2-\theta_1)} (1-t/G)^{x(\theta_2\vartheta_2-\theta_1\vartheta_1)-1} T_{\fth}^{-2} \quad \text{and} \\
	F_2& := & x(\theta_2\vartheta_2-\theta_1\vartheta_1)T_{\fth_1}T_{\fth_2}+(\vartheta_2\theta_2 x+1)\vartheta_2 T_{\fth_1} -(\vartheta_1\theta_1 x+1)\vartheta_1 T_{\fth_2}.
\end{eqnarray*}
The first factor is only zero for $t=G$. Setting the second factor equal to zero results in $T_{\fth_2}\vartheta_1 [ T_{\fth_1}\theta_1 x+(\vartheta_1\theta_1 x+1)]
=T_{\fth_1}\vartheta_2  [ T_{\fth_2}\theta_2 x+(\vartheta_2\theta_2 x+1)]$.
The parameter $t$ is only included in term $\log(1-t/G)$. Define $y:=\log(1-t/G)$, to see that two bivariate polynomials of equal grade are equal. This is only the case, if (and only if) their coefficients coincide. A comparison of the coefficients of the absolute terms results into
\begin{align*}
	-\vartheta_1+\vartheta_2\vartheta_1=-\vartheta_2+\vartheta_2 \vartheta_1 \qquad \Longleftrightarrow \qquad \vartheta_1=\vartheta_2.
\end{align*}
Equating $\vartheta_1=\vartheta_2$ and comparing the coefficients with the variable $x$ results in
$\theta_1(\vartheta_2-2\vartheta_2^2+2\vartheta_2^3)=\theta_2(\vartheta_2-2\vartheta_2^2+2\vartheta_2^3)$
from which $\theta_1=\theta_2$ can be concluded for $\vartheta_2 \neq 0$. In the case of independence, i.e. for $\vartheta_1=\vartheta_2=0$, it would result in
\begin{align*}
	\frac{f_{\fth_2}(x,t)}{f_{\fth_1}(x,t)}=\frac{\theta_2}{\theta_1} e^{-x(\theta_2-\theta_1)}=1
\end{align*}
and did allow to conclude $\theta_1=\theta_2$ directly. \qed 

\subsection{Display of determinant in Assumption \ref{A5negdef}} \label{xy}

\nopagebreak 

\begin{figure}[H]
	\centering
	\begin{minipage}{0.48\linewidth}
		\includegraphics[width=7cm]{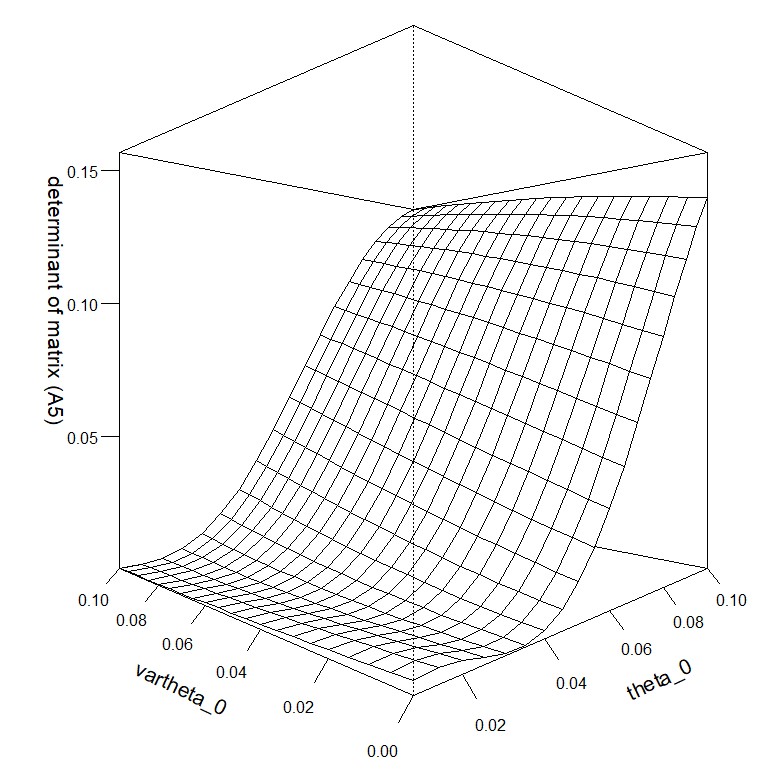}
	\end{minipage}
	\hspace{0.1cm}
	\begin{minipage}{0.48\linewidth}
		\includegraphics[width=7cm]{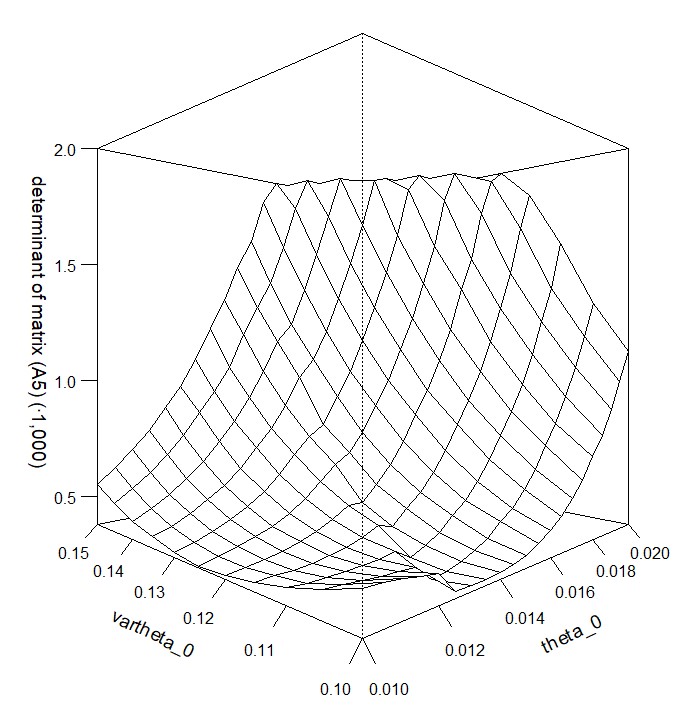}
	\end{minipage}
	\caption{Determinant of $\mathbb{E}_{\fth_0}[\dot{\psi}_{\fth_0,1}(X_1,T_1)]$ for relevant range of parameter space in the application in Section \ref{empexa} (left) and scaled up to area of smallest determinant (right)}
	\label{detA5graph}
\end{figure}

\subsection{Truncated population model} \label{idcond}

For two values $\boldsymbol{\theta}_1,\boldsymbol{\theta}_2 \in \Theta$, it is $KL(\tilde{f}_{\boldsymbol{\theta}_1},\tilde{f}_{\boldsymbol{\theta}_2})=0$ if and only if $\tilde{f}_{\boldsymbol{\theta}_1}=\tilde{f}_{\boldsymbol{\theta}_2}$, i.e. if
\[
\frac{f_{\boldsymbol{\theta}_2}}{f_{\boldsymbol{\theta}_1}} = \frac{\alpha_{\boldsymbol{\theta}_1}}{\alpha_{\boldsymbol{\theta}_2}}
\]
for all $(x,t) \in D$. By inserting \eqref{e4} one arrives again at equality \eqref{dichtquot}, only with a right hand side as $\alpha_{\boldsymbol{\theta}_1}/\alpha_{\boldsymbol{\theta}_2} $ instead of $1$. Taking derivatives with respect to $t$  has the same result as for equality \eqref{dichtquot} because $\alpha_{\boldsymbol{\theta}}$ does not depend on $t$. So that the remaining arguments for identification are as in Section \ref{identsrs}.

\subsection{Proof of Lemma \ref{lem1}} \label{prooflem1}

One may perform the integrations, introduced by the expectations, and prove equality to zero \cite[see e.g.][Appendix A]{weiswied2021}. Instead here we prove equality to zero of the integrand. We verify each coordinate separately, for \eqref{e7} we write $\mathbb{E}_{\fth_0}[\psi_{\fth_0,1}(X_1,T_1)]=s_1+s_2-s_3$ with
\begin{equation*}
	s_1:=\mathbb{E}_{\fth_0}(\mathds{1}_{[T_1,T_1+s]}(X_1){\theta_0}^{-1}), \; 	s_3:=\frac{\partial\alpha_{\fth_0}}{\partial \theta_0}
\end{equation*}
and
\begin{align*}
	s_2&:=\mathbb{E}_{\fth_0}\Bigg[\mathds{1}_{[T_1,T_1+s]}(X_1)\Big(X_1(\vth_0\log(1-T_1/G)-1)+\\
	&\qquad\qquad\frac{\vth_0 X_1(\vth_0 \log(1-T_1/G)-1)}{(\vth_0 \theta X_1+1)(\vth_0\log(1-T_1/G)-1)+\vth_0}\Big)\Bigg].
\end{align*}
Represent now $\mathbb{E}_{\fth_0}$ with \eqref{e4} for $s_1$ and $s_2$. For $s_3$ use the representation \eqref{alpha0} for $\alpha_{\fth_0}$ and interchange partial derivation and integration. Finally $s_1=\int_D -G^{-1} e^{-\theta_0 x} (1-t/G)^{\vth_0 \theta_0 x}  [(\vth_0 \theta_0 x +1)(\vth_0 \log(1-t/G)-1)+\vth_0] \dt(x,t)$, $s_2=\int_D -x(\vth_0\log(1-t/G)-1)(\theta_0/G) e^{-\theta_0 x} (1-t/G)^{\vth_0 \theta_0 x} [(\vth_0 \theta_0 x +1)(\vth_0 \log(1-t/G)-1)+2\vth_0] \dt(x,t)$ and
$s_3=\int_D -G^{-1} e^{-\theta_0 x} (1-t/G)^{\vth_0 \theta_0 	x}\{ \theta_0 \vth_0 x(\vth_0\log(1-t/G)-1)+
[1+\theta_0 x (\vth_0 \log(1-t/G)-1)][(\vth_0 \theta_0 x +1)(\vth_0 \log(1-t/G)-1)+\vth_0]\} \dt(x,t) =\int_D -G^{-1} e^{-\theta_0 x} (1-t/G)^{\vth_0 \theta_0 	x}\{ \theta_0 \vth_0 x(\vth_0\log(1-t/G)-1)+\vth_0 +(\vth_0 \theta_0 x +1)(\vth_0 \log(1-t/G)-1)+\vth_0 \theta_0 x(\vth_0 \log(1-t/G)-1)+ \theta_0 x (\vth_0\theta_0 x +1)(\vth_0 \log(1-t/G)-1)^2\} \dt(x,t)$.
Without $s_2$ it is $s_1-s_3=\int_D G^{-1} e^{-\theta_0 x} (1-t/G)^{\vth_0 \theta_0 	x} \{ \theta_0 \vth_0 x(\vth_0\log(1-t/G)-1)
+ \vth_0 \theta_0 x(\vth_0 \log(1-t/G)-1)+
\theta_0 x (\vth_0\theta_0 x +1)(\vth_0 \log(1-t/G)-1)^2\} \dt(x,t)$.
Adding $s_2$ finally leads to $\mathbb{E}_{\fth_0}[\psi_{\fth_0,1}(X_1,T_1)]=0$. For the second coordinate $\mathbb{E}_{\fth_0}[\psi_{\fth_0,2}(X_1,T_1)]$ of \eqref{eq:1}, arguments are similar. \qed

\section{Proof of Theorem \ref{satznormoffen}}\label{appnorm} 

As function in $\fth$, the rational functions in \eqref{eq:1} are two times continuously differentiable, because the denominator has no zeros for $\fth \in \Theta$ and $(X_i,T_i) \in D$. Furthermore, by Lemma \ref{propalpha}, the selection probability  $\alpha_{\fth}$, as function of $\fth$, is three times partially differentiable. The function $\fth \mapsto \psi_{\fth}(x,t)$ is hence, as composition of two times continuously differentiable functions, for almost all $(x,t) \in S$, as well two times continuously differentiable.
That $\mathbb{E}_{\fth_0}\big[\dot{\psi}_{\fth_0}(X_1,T_1)\big]$ is invertible follows from Assumption \ref{A5negdef}.

For the second condition \citep[Theorem 5.41]{vaart1998} note that by Lemma \ref{e18}, expanding the square for $g^2(x,t)$ results in 
$\psi_{\fth_0,1}^2(x,t)+\psi_{\fth_0,2}^2(x,t)\leq K_{\varepsilon}^2\{1+2[1-\log(1-t/G)]+[1-\log(1-t/G)]^2\}
=K_{\varepsilon}^2\{4-4\log(1-t/G)]+\log(1-t/G)^2\}$.
Note that $\mathbb{E}_{\fth_0} [1-\log(1-T_1/G)]=2$. By substituting $x=1-t/G$ and partially integrating with $u(x)=\log(x)^2$ and $v(x)=1$ one has $\mathbb{E}_{\fth_0} [(\log(1-T_1/G))^2]=\int_0^1 (\log(x))^2 \dt x=-\int_0^1 2\log(x) \dt x=2$ 
and hence $\mathbb{E}_{\fth_0}[\Vert \psi_{\fth_0}(X_1,T_1)\Vert^2]\leq 10 K_{\varepsilon}^2<\infty$.
The lengthy derivations of integrable bounds to the second partial derivatives of $\psi_{\fth}$ use similar arguments and are suppressed here. \qed

In order to proof Formula \eqref{ime}, recall the representation of  $\tilde{f}_{\fth}(x,t)$ from Section \ref{profmod} and consider the `anti-clockwise' model (Figure \ref{exa})
to derive $\tilde{\mathbb{E}}_{\fth_0}(\psi_{\fth_0}(\tilde{X}_1,\tilde{T}_1)\psi_{\fth_0}(\tilde{X}_1,\tilde{T}_1)')= \mathbb{E}_{\fth_0}(\psi_{\fth_0}(X_1,T_1)\psi_{\fth_0}(X_1,T_1)')/\alpha_{\fth_0}$ as well as
$-\tilde{\mathbb{E}}_{\fth_0}[\dot{\psi}_{\fth_0}(\tilde{X}_1,\tilde{T}_1)]=- \mathbb{E}_{\fth_0}[\dot{\psi}_{\fth_0}(X_1,T_1)]/\alpha_{\fth_0}$.

For the sake of brevity denote $\tilde{f}_{\fth_0}=\tilde{f}_{\fth_0}(\tilde{X}_1,\tilde{T}_1)$. Using matrix differentiation let $\partial \tilde{f}_{\fth}/\partial \fth:=(\partial \tilde{f}_{\fth}/\partial \theta,\partial \tilde{f}_{\fth}/\partial \vth)'$ be the gradient of the function $\tilde{f}_{\fth}$ and $\partial \tilde{f}_{\fth}/\partial \fth'$ its transpose. Accordingly, denote by $\partial^2 \tilde{f}_{\fth}/(\partial \fth \partial \fth')$ the Hessian. 
Now it holds
\begin{align*}
	\tilde{\mathbb{E}}_{\fth_0}[\dot{\psi}_{\fth_0}(\tilde{X}_1,\tilde{T}_1)]
	&=\tilde{\mathbb{E}}_{\fth_0}\left[\frac{\partial^2 \log \tilde{f}_{\fth_0}}{\partial \fth\partial \fth'}\right]
	=\tilde{\mathbb{E}}_{\fth_0}\left[\frac{\partial}{\partial \fth} \left(\frac{1}{\tilde{f}_{\fth_0}}\frac{\partial \tilde{f}_{\fth_0}}{\partial \fth'}\right)\right]\\
	&=\tilde{\mathbb{E}}_{\fth_0}\left[-\frac{1}{\tilde{f}_{\fth_0}^2}\frac{\partial \tilde{f}_{\fth_0}}{\partial \fth}\frac{\partial \tilde{f}_{\fth_0}}{\partial \fth'}+\frac{1}{\tilde{f}_{\fth_0}}\frac{\partial^2 \tilde{f}_{\fth_0}}{\partial \fth \partial \fth'}\right]\\
	&=-\tilde{\mathbb{E}}_{\fth_0}\left[\left(\frac{1}{\tilde{f}_{\fth_0}}\frac{\partial \tilde{f}_{\fth_0}}{\partial \fth}\right)\left(\frac{1}{\tilde{f}_{\fth_0}}\frac{\partial \tilde{f}_{\fth_0}}{\partial \fth'}\right)\right]
	+\tilde{\mathbb{E}}_{\fth_0}\left[\frac{1}{\tilde{f}_{\fth_0}}\frac{\partial^2 \tilde{f}_{\fth_0}}{\partial \fth \partial \fth'}\right]\\
	&=-\tilde{\mathbb{E}}_{\fth_0}\left[\left(\frac{\partial \log \tilde{f}_{\fth_0}}{\partial \fth}\right)\left(\frac{\partial \log \tilde{f}_{\fth_0}}{\partial \fth'}\right)\right]
	+ \int_S \frac{\partial^2 \tilde{f}_{\fth_0}(x,t)}{\partial \fth \partial \fth'} \dt (x,t).
\end{align*}
The IME for the `anti-clockwise' model is achieved when the second summand vanishes, i.e. when all entries of the matrix equal zero. We only depict the calculation for the first component here. As a further simplification, denote by  $\dot{f}_{\fth_0}$ the first, respectively by $\ddot{f}_{\fth_0}$ the second, derivative with respect to $\fth$'s first  coordinate, $\theta$. Similarly for the selection probability $\alpha_{\fth_0}$. On $D$ holds
$\partial \tilde{f}_{\fth_0}/\partial \theta=(\dot{f}_{\fth_0} \alpha_{\fth_0}-\dot{\alpha}_{\fth_0} f_{\fth_0})/\alpha_{\fth_0}^2$
as well as 
\begin{equation*}
	\frac{\partial^2 \tilde{f}_{\fth_0}}{\partial \theta^2}=\frac{\ddot{f}_{\fth_0} \alpha_{\fth_0}-\ddot{\alpha}_{\fth_0} f_{\fth_0}}{\alpha_{\fth_0}^2}-\frac{2\dot{\alpha}_{\fth_0} (\dot{f}_{\fth_0} \alpha_{\fth_0}-\dot{\alpha}_{\fth_0} f_{\fth_0})}{\alpha_{\fth_0}^3}.
\end{equation*}
Due to $\int_D f_{\fth_0}(x,t) \dt (x,t)=\alpha_{\fth_0}$ one has for the integral  
\begin{equation*}
	\int_D \frac{\partial^2 \tilde{f}_{\fth_0}(x,t)}{\partial \theta^2}\dt (x,t)=\int_D \frac{\ddot{f}_{\fth_0}(x,t)}{\alpha_{\fth_0}} \dt (x,t)-\frac{\ddot{\alpha}_{\fth_0}}{\alpha_{\fth_0}}-\frac{2\dot{\alpha}_{\fth_0}}{\alpha_{\fth_0}^2}\int_D \dot{f}_{\fth_0}(x,t) \dt (x,t)+\frac{2\dot{\alpha}_{\fth_0}^2}{\alpha_{\fth_0}^2}.
\end{equation*}
Lemma \ref{propalpha} ensures changeability of differentiation and integration, i.e. it is  $\int_D \dot{f}_{\fth_0}(x,t) \dt (x,t)=\dot{\alpha}_{\fth_0}$ and $\int_D \ddot{f}_{\fth_0}(x,t) \dt (x,t)=\ddot{\alpha}_{\fth_0}$. The value of the integral for the second derivative of the density $\tilde{f}_{\fth_0}$ is hence zero. This can be similarly shown for the other three components of the above matrix. 

As result one has the IME for the `anti-clockwise' model and, as described earlier thereof for the model of dependent truncation.    \qed

\newpage
\section{Simulation results for Section \ref{simressec}} \label{ressec}

\begin{longtable}{ccccccc}
	\caption{Simulation result for bias and variance of $\hat{\theta}$ and $\hat{\vartheta}$ as zeros of \eqref{eq:1} for different scenarios (subscript $n$ is omitted for the sake of brevity)} \label{sims} \\ 	\toprule
	& \multicolumn{3}{c}{\underline{$G=24,\, s=3$}} & \multicolumn{3}{c}{\underline{$G=24,\, s=48$}}\\ 
	\multicolumn{2}{c}{$\theta_0=0.05, \, \vth_0=0.001$} & \multicolumn{2}{c}{$\alpha_{\fth_0}=0.0648$} & &  \multicolumn{2}{c}{$\alpha_{\fth_0}=0.5294$} \\
	\toprule
	$n$ & $10^3$ & $10^4$ & $10^5$ &  $10^3$ & $10^4$ & $10^5$ \\
	\midrule[1.6pt]
	Bias$(\hat{\theta})$ 	&	0.00403 	& 	0.001456 	& 	0.000446 	& 		0.000136 	& 	-0.000011  						& 	0.0000044 \\ 
	Bias$(\hat{\vth})$ 		&	0.12373 	& 	0.047852 	& 	0.017451 & 		0.024781	&  	{0.005959}	&  	0.0013721 \\ 
	Var$(\hat{\theta})$		& 	0.00026 	&	0.000026 	& 	0.000003 	& 		0.000010	& 	{0.000001}	&  	0.0000001 \\
	Var$(\hat{\vth})$		&	0.02513 	& 	0.003987	& 	0.000610 	& 		0.001227	& 	{0.000102}	& 	0.0000108 \\
	\midrule[1.6pt]
	\multicolumn{2}{c}{$\theta_0=0.05, \, \vth_0=0.01$} & \multicolumn{2}{c}{$\alpha_{\fth_0}=0.0808$} & & \multicolumn{2}{c}{$\alpha_{\fth_0}=0.5286$}\\
	\toprule
	$n$ & $10^3$ & $10^4$ & $10^5$   
	& $10^3$ & $10^4$ & $10^5$   \\
	\midrule[1.6pt]
	Bias$(\hat{\theta})$ 	& 	0.00241 	& 	0.001411  	& 	0.001110 	& 	0.00002 	& 	-0.000106  						& -0.0000881 \\ 
	Bias$(\hat{\vth})$ 		&	0.13245 	& 	0.057141 	&  	0.038426 	&	0.02023		&  	{0.003510}	&  -0.0000125 \\ 
	Var$(\hat{\theta})$		& 	0.00025		& 	0.000027 	& 	0.000003 	& 	0.00001		& 	{0.000001}	&  {0.0000001}\\
	Var$(\hat{\vth})$		& 	0.02961		& 	0.004975	& 	0.001038	&	0.00150		&	{0.000198}	 & {0.0000312}\\		\midrule[1.6pt]	
	\multicolumn{2}{c}{$\theta_0=0.1, \, \vth_0=0.001$} & \multicolumn{2}{c}{$\alpha_{\fth_0}=0.0982$} & & \multicolumn{2}{c}{$\alpha_{\fth_0}=0.37575$} \\
	\toprule
	$n$ & $10^3$ & $10^4$ & $10^5$  & $10^3$ & $10^4$ & $10^5$ \\
	\midrule[1.6pt]
	Bias$(\hat{\theta})$ 	& 	0.00075 	& 	-0.00010 						& -0.000192 	& 	0.00035 	& 	0.000011  	& -0.0000253\\ 
	Bias$(\hat{\vth})$ 		& 	0.07880 	&  	{0.02142}	&  {0.005680}	& 	0.01883		&  	0.004558	&  {0.0007957}	\\ 
	Var$(\hat{\theta})$		& 	0.00028		& 	{0.00003}	&  {0.000003}	&  	0.00002		& 	0.000002	&  {0.0000002}	\\
	Var$(\hat{\vth})$		& 	0.01147		& 	{0.00099}	& {0.000087}	& 	0.00083		& 	0.000064	& 	{0.0000063}\\		\midrule[1.6pt]	
	\multicolumn{2}{c}{$\theta_0=0.1, \, \vth_0=0.01$} & \multicolumn{2}{c}{$\alpha_{\fth_0}=0.0978$} & &  \multicolumn{2}{c}{$\alpha_{\fth_0}=0.37574$}\\
	\toprule
	$n$ & $10^3$ & $10^4$ & $10^5$ 	& $10^3$ & $10^4$ & $10^5$ \\
	\midrule[1.6pt]
	Bias$(\hat{\theta})$ 	& 	0.00015 	& 	-0.00095  					& -0.000884 	& 	0.00011	 	& 	-0.00018  						& -0.0001896 \\ 
	Bias$(\hat{\vth})$ 		& 	0.07874		&  	{0.01768}	&  {0.006281}& 	0.01086		&  	-0.002902						&  -0.0061084\\ 
	Var$(\hat{\theta})$		&  	0.00025		& 	{0.00003}	&  {0.000003}&  	0.00002		& 	{0.000002}	&  {0.0000002}\\
	Var$(\hat{\vth})$		& 	0.01340		& 	{0.00117}	& {0.000195}&	0.00078		& 	{0.000080}	& {0.0000122}\\ 	\midrule[1.6pt] 	\midrule[1.6pt]
	&		\multicolumn{3}{c}{\underline{$G=48, \, s=3$}} & \multicolumn{3}{c}{\underline{$G=24, \, s=2$}} \\
	\multicolumn{2}{c}{$\theta_0=0.05, \, \vth_0=0.001$} &  \multicolumn{2}{c}{$\alpha_{\fth_0}=0.0528$} & &  	\multicolumn{2}{c}{$\alpha_{\fth_0}=0.0554$} \\
	\toprule
	$n$ & $10^3$ & $10^4$ & $10^5$ 	& $10^3$ & $10^4$ & $10^5$ \\
	\midrule[1.6pt]
	Bias$(\hat{\theta})$ &	0.00104 & -0.00016  					& -0.000031 	& 0.00447	& 0.00137	& 0.000822  \\ 
	Bias$(\hat{\vth})$ 	& 	0.12574	&  {0.03096}	& {0.008380}	& 0.15437	& 0.05441	& 0.023376\\ 
	Var$(\hat{\theta})$	&  	0.00014	&	{0.00001} & {0.000001}	& 0.00032	& 0.00004	& 0.000005\\
	Var$(\hat{\vth})$	& 	0.02701	&	{0.00200} & {0.000167}	& 0.03737	& 0.00514	& 0.000968\\ 	\midrule[1.6pt]
	\multicolumn{2}{c}{$\theta_0=0.05, \, \vth_0=0.01$} & \multicolumn{2}{c}{$\alpha_{\fth_0}=0.0525$} & &   \multicolumn{2}{c}{$\alpha_{\fth_0}=0.0552$}\\
	\toprule
	$n$ & $10^3$ & $10^4$ & $10^5$ & $10^3$ & $10^4$ & $10^5$  \\
	\midrule[1.6pt]
	Bias$(\hat{\theta})$ & 0.00055 	& -0.00035  				  & -0.000379 	& 0.00371	& 0.00184   & 0.001204 \\ 
	Bias$(\hat{\vth})$ 	 & 0.11406	& {0.02914} & {0.007752} 	& 0.15143	& 0.06513	& 0.037434\\ 
	Var$(\hat{\theta})$	 & 0.00012 	& {0.00001} & {0.000001} 	& 0.00033	& 0.00004	& 0.000005 \\
	Var$(\hat{\vth})$	 & 0.02528	& {0.00244} & {0.000298}	& 0.03774	& 0.00663	& 0.001365\\		\midrule[1.6pt]
	\multicolumn{2}{c}{$\theta_0=0.1, \, \vth_0=0.001$} &  \multicolumn{2}{c}{$\alpha_{\fth_0}=0.0535$} & &   \multicolumn{2}{c}{$\alpha_{\fth_0}=0.0686$}\\
	\toprule
	$n$ & $10^3$ & $10^4$ & $10^5$ & $10^3$ & $10^4$ & $10^5$  \\
	\midrule[1.6pt]
	Bias$(\hat{\theta})$ & 0.00229 	& 0.00018 	 & -0.000015 & 0.00231	& -0.00008  					& -0.000177 \\ 
	Bias$(\hat{\vth})$ 	 & 0.11244	& 0.02525	 & {0.006045}	& 0.10145	& {0.02562}	& {0.008062}\\ 
	Var$(\hat{\theta})$	 & 0.00030  & 0.00002	 & {0.000002} 	& 0.00042	& {0.00004}	& {0.000004}\\
	Var$(\hat{\vth})$	 & 0.02188	& 0.00141	 & {0.000104}	& 0.01817	& {0.00150}	& {0.000138}\\ \midrule[1.6pt]
	\multicolumn{2}{c}{$\theta_0=0.1, \, \vth_0=0.01$} &  \multicolumn{2}{c}{$\alpha_{\fth_0}=0.0534$} & &  \multicolumn{2}{c}{$\alpha_{\fth_0}=0.0684$}\\
	\toprule
	$n$ & $10^3$ & $10^4$ & $10^5$ & $10^3$ & $10^4$ & $10^5$  \\
	\midrule[1.6pt]
	Bias$(\hat{\theta})$ & 0.00275	 & -0.00023  					& -0.000580 & 0.00126	 & -0.00046  					& -0.000618\\ 
	Bias$(\hat{\vth})$ 	 & 0.11327	 & {0.01970}	& -0.000619& 0.10149	 & {0.02399} 	&  {0.007045}\\ 
	Var$(\hat{\theta})$	 & 0.00030	 & {0.00002}	& {0.000002} 	& 0.00039	 & {0.00004}	&  {0.000004}\\
	Var$(\hat{\vth})$	 & 0.02576	 & {0.00154}	& {0.000123}	& 0.01981	 & {0.00182}	& {0.000249}\\ \bottomrule[1.6pt]		
\end{longtable}

\end{document}